\documentclass[11pt]{article}
\usepackage{amsmath}
\usepackage{mathrsfs}
\usepackage{latexsym}
\usepackage[latin1]{inputenc}
\usepackage{amssymb}
\usepackage{amsfonts}
\usepackage{epsfig}
\usepackage{theorem}
\usepackage{array}
\usepackage[round]{natbib}
\usepackage{graphicx}

\newtheorem{theorem}{Theorem}[section]

\newtheorem{lemma}{Lemma}[section]

\newtheorem{definition}{Definition}[section]

{\theorembodyfont{\rmfamily}
\newtheorem{example}{Example}
\newtheorem{remark}{Remark}
}

\numberwithin{equation}{section}

\def\qed{\hfill{$\Box$} \\}

\def\expect{{\mathbb  E}}
\def\Pr{{\mathbb P}}
\def\RR{{\mathbb R}}
\def\real{{\mathbb  R}}
\def\nat{{\mathbb  N}}

\def\eqdef{\triangleq}
\def\eqd{\stackrel{d}{=}}
\def\ind{{\bf 1}}

\renewcommand{\leq}{\leqslant}
\renewcommand{\geq}{\geqslant}

\begin{document}
\thispagestyle{plain} \pagestyle{plain}
\date{{\small January 2007; revised June 2007}}
\author{Predrag R. Jelenkovi\'c  \;\;  Xiaozhu Kang \;\; Jian Tan \\
\small Department of Electrical Engineering \\
\small Columbia University, New York, NY 10027 \\
\small  \{predrag, xiaozhu, jiantan\}@ee.columbia.edu\\
\small Tel.: (212) 854 8174  \;\; Fax:  (212) 932 9421 }

\renewcommand{\footnote}{}
\title{\bf{Heavy-Tailed Limits for Medium Size Jobs and Comparison Scheduling}\footnote{This work is supported by NSF Grant 0615126.}
 }

\maketitle \normalsize

\renewcommand{\thefootnote}{}

\begin{abstract}
We study the conditional sojourn time distributions of processor
sharing (PS), foreground background processor sharing (FBPS) and
shortest remaining processing time first (SRPT) scheduling
disciplines on an event where the job size of a customer arriving in
stationarity is smaller than exactly $k\ge 0$ out of the preceding
$m\ge k$ arrivals. Then, conditioning on the preceding event, the
sojourn time distribution of this newly arriving customer behaves
asymptotically the same as if the customer were served in isolation
with a server of rate $(1-\rho)/(k+1)$ for PS/FBPS, and $(1-\rho)$
for SRPT, respectively, where $\rho$ is the traffic intensity.
Hence, the introduced notion of conditional limits allows us to
distinguish the asymptotic performance of the studied schedulers by
showing that SRPT exhibits considerably better asymptotic behavior
for relatively smaller jobs than PS/FBPS.

Inspired by the preceding results, we propose an approximation to
the SRPT discipline based on a novel adaptive job grouping mechanism
that uses relative size comparison of a newly arriving job to the
preceding $m$ arrivals. Specifically, if the newly arriving job is
smaller than $k$ and larger than $m-k$ of the previous $m$ jobs, it
is routed into class $k$. Then, the classes of smaller jobs are
served with higher priorities using the static priority scheduling.
The good performance of this mechanism, even for a small number of
classes $m+1$, is demonstrated using the asymptotic queueing
analysis under the heavy-tailed job requirements. We also discuss
refinements of the comparison grouping mechanism that improve the
accuracy of job classification at the expense of a small additional
complexity.

\vspace{0.6cm}

\noindent \textbf{Keywords}: Comparison scheduling, scalability,
fairness, adaptive thresholds, M/G/1 queue, processor sharing,
shortest remaining processing time first, foreground background
processor sharing, asymptotic analysis, heavy tails, medium size
jobs
\end{abstract}

\newpage
\section*{Introduction}
It has been widely recognized that heavy-tailed distributions are
suitable for modeling job sizes in information service networks,
e.g., see \cite{JEM00b,JEMO03} and the references therein. For
heavy-tailed distributions, large jobs appear much more frequently
than for the light-tailed ones, which imposes very different
constraints in terms of optimizing the scheduling process as
compared to the light-tailed scenarios. In particular, schedulers
that may assign the server exclusively to a very large job, e.g.,
first come first serve (FIFO) discipline, can cause very large
delays and, in general, suboptimal performance, as shown by
\cite{Ana99}.

Hence, most of the practical schedulers utilize either the processor
sharing (PS) and foreground background processor sharing (FBPS)
disciplines because of their inherent fairness, or the shortest
remaining processing time first (SRPT) discipline because of its
known optimality under quite general conditions. In particular,  it
was shown by \cite{SCH68} that SRPT minimizes the number of
customers in the G/G/1 queue over all work-conserving disciplines.
For early references on these and other scheduling disciplines see
\cite{Kle76,Wol89} and the references therein. Recently, the
performance of these disciplines was revisited in the context of
heavy tails; for a recent survey see \cite{BBQ03}. For practical
applications of SRPT-based scheduling to improving Web server
performance see \cite{harsba03,rawat03swift}; also, for recent
studies that are applying FBPS to reducing the latency of short TCP
flows see \cite{RaoUVB04,RaiBU05}.

It is well known that the sojourn time distributions under PS, FBPS
and SRPT scheduling disciplines are asymptotically equivalent for
power law distributions (more precisely, regularly or intermediately
regularly varying distributions). This was originally proved by
\cite{NQ00} and then later studied for regularly varying
distributions in Theorems~2.2, 2.5 and 2.6 of \cite{BBQ03}; see also
Theorem~2.1 of \cite{JEMO03} and Theorem~1 of \cite{Predrag02}. In
other words, for large jobs, the waiting time does not depend on the
choice of a specific scheduling discipline among PS, FBPS and SRPT.

In this paper, we introduce a new notion of conditional waiting
time distribution which allows us to refine and distinguish the
performance of PS/FBPS and SRPT schedulers for medium size jobs.
Informally, our first main result, stated in
Theorem~\ref{theorem:main}, shows that even the relatively smaller
jobs receive asymptotically the same residual capacity $1-\rho$ as
the larger ones for SRPT discipline, while, for PS/FBPS
schedulers, these smaller jobs share the residual capacity equally
with the larger jobs in the system. Hence, it appears that SRPT
provides much better and more uniform performance over a wide
range of time scales. Furthermore, the performance improvement for
conditionally smaller jobs is not achieved at the expense of
larger jobs, i.e., SRPT is not only efficient but fair as well,
which is in line with similar recent findings in the context of
mean value analysis by \cite{bansal01,admo03}. To this end, we
would like to point out that contrary to our findings, in the
light-tailed context, it was shown by \cite{RS01} that FIFO is
optimal in terms of maximizing the decay rate of the waiting time
distribution over all work conserving disciplines. For more recent
results on the light-tailed asymptotic analysis see \cite{NZ06}
and the references therein.

Overall, using the SRPT scheduling is beneficial for a broad range
of conditions and applications. However, as discussed in one of
the very first papers on SRPT by \cite{SCMI66}, this discipline
may be quite difficult to implement. Clearly, its complicated
preemptive nature requires keeping track of the remaining
processing times for all jobs in the queue which may be
prohibitive for systems with large job volumes, e.g., Web servers.
In addition, \cite{SCMI66} show that the expected number of
preemptions per job is proportional to the load of the system,
which can be quite large. Hence, even as early as 1966, it was
recognized by \cite{SCMI66} that one should try to approximate
SRPT with less complex schedulers. The most apparent option, as
suggested by \cite{SCMI66}, is to design a threshold-based static
priority approximation to SRPT. Basically, the idea is to select a
fixed number of thresholds $m$ and then group jobs into $m+1$
classes depending on which pair of thresholds a job size happens
to fall between. Then, these classes are served according to the
static priority discipline with higher priorities assigned to
classes with smaller jobs. Since then, there has been a lot of
work on threshold-based scheduling policies.  For example, it was
shown by \cite{NBDG06} that even with a single threshold, one can
obtain the performance comparable to SRPT up to a constant factor
in terms of the mean sojourn time for {M/M/1} queue as well as for
{M/G/1} queue with finite variance Pareto service distribution.

Although it is encouraging that one can achieve a provably very good
approximation of M/G/1/SRPT queue even with a very small number of
static thresholds (only one in the paper by \cite{NBDG06}), these
solutions are likely not to perform well in practice since the
traffic characteristics are often nonstationary, highly correlated
(long range dependent) and very bursty (e.g., batch arrivals, etc);
see \cite{PAW00,SYL99}. In order to overcome these difficulties, we
propose a novel adaptive job classification (grouping) mechanism
that is based on relative size comparison of a newly arriving job to
the previous $m$ arrivals; this scheduler is inspired by our
conditional limit results. Specifically, if an arriving job is
smaller than $k$ and larger than $m-k$ of the previous $m$ jobs, it
is routed into class $k$. We also discuss refinements of the
comparison grouping mechanism that improve the accuracy of the
classification for both light-tailed and correlated job arrivals at
the expense of a small (fixed) additional complexity in
Subsection~\ref{sec:ref}.

The good performance of our comparison classification mechanism is
demonstrated using the asymptotic queueing analysis under
heavy-tailed job sizes in Section~\ref{sec:qana}. First, in
Subsection~\ref{ss:qi} we study the queueing behavior of a class $k$
process in isolation and show that the workload distribution decays
faster for larger $k$. More precisely, for regularly varying (power
law) service distribution, the tail of the workload distribution
$\Pr[W^{(k)}>x]$ of a class $k$ process, as implied by
Theorem~\ref{theorem:isolate}, is of the order of
$x(\Pr[B>x])^{k+1}$, where $B$ is the service requirement of a
typical job before the comparison splitting. Hence, our comparison
splitting procedure provides a proper ordering of jobs. Furthermore,
in Subsection~\ref{ss:sp} we study the joint queueing behavior of
all classes under the static priority (SP) discipline, with higher
priorities assigned to classes with ``smaller'' jobs.
Theorem~\ref{theorem:static} shows that the workload distribution of
a class with a smaller index $k$ (i.e., larger jobs) has the same
queueing behavior as if it were served in isolation with the system
capacity reduced by the mean arrival rates of the classes with
smaller jobs. Roughly speaking, this is a similar behavior as seen
in Theorem~\ref{theorem:main} for the SRPT discipline and, thus, the
SP scheduling with our comparison splitting should provide a
reasonable approximation to the SRPT discipline. Furthermore, in
regard to the analysis, we would like to point out that the main
technical difficulty is that the split processes are individually
and mutually correlated. This statistical correlation makes other
types of analyses, outside of the heavy-tailed context, possibly
difficult.

In addition, we would like to point out that a preliminary version
of this paper has appeared earlier in  \cite{JKTSIG07} as part of
the conference proceedings, which contains sketches of the proofs
as well as the extensive simulation experiments.
 Those experiments demonstrated the
good performance of our adaptive scheduler that, in particular,
outperforms the static threshold policies when the arrival processes
are statistically correlated and time varying. However, in contrast
to the previous focus on simulations in \cite{JKTSIG07}, this paper
provides  the rigorous details of the proofs.

The rest of this paper is structured as follows. In the next
section, we introduce the new notion of conditional waiting time
distribution that refines and differentiates the performance of
PS/FBPS and SRPT schedulers for medium size jobs. Based on the
conditional asymptotic result of the sojourn time distribution
(stated in Theorem~\ref{theorem:main}), we propose a novel
comparison grouping scheme and its refined version in
Section~\ref{sec:comparison}. To demonstrate its good performance,
we conduct the asymptotic queueing analysis under heavy-tailed job
sizes in Section~\ref{sec:qana}. In the end, Section~\ref{sec:con}
summarizes our contributions.

\section{Heavy-Tailed Limits for Medium Size Jobs with Popular Schedulers}
\label{s:limit}
\subsection{Definitions and Preliminary Results} In this section we
introduce the necessary notation and describe the existing and
preliminary results. Let $B_i$ and $V_i$ denote the job size and
the waiting time of the customer arriving at time $T_i$,
respectively, where $\{B_i\}_{i>-\infty}$ are i.i.d. random
variables. The arrival points $\{T_i\}_{i>-\infty}$ are assumed to
be Poisson with rate $\lambda$ and independent of job requirements
$\{B_i\}_{i>-\infty}$. Hence, without loss of generality, in view
of the PASTA property, we set $T_0=0$. The waiting time of a
customer is defined as the amount of time between its arrival and
departure, also referred to as sojourn time in the queuing
literature. To present our main results, we need the following
definitions.
\begin{definition}
\label{def:ir}
 A nonnegative random variable $X$ or its distribution function (d.f.) $F$ is
called intermediately regularly varying, $X \in {\cal IR}$, if
\begin{equation*}
\lim_{\eta \uparrow 1} \varlimsup_{x \rightarrow \infty}
\frac{\Pr[X>\eta x]} {\Pr[X>x]} = 1.
\end{equation*}
\end{definition}
Regularly varying distributions ${\cal R}_\alpha$ are the
best-known examples from ${\cal IR}$.
\begin{definition}
\label{def:r}
 A nonnegative random variable $X$ or its d.f. $F$ is
called regularly varying with index $\alpha$, $X\in {\cal
R}_\alpha$ $(F\in{\cal R}_\alpha)$, if
$$F(x)=1-\frac{l(x)}{x^\alpha},\;\; \alpha\ge 0,$$
where $l(x)$: $\RR_+\rightarrow \RR_+$ is slowly varying, i.e.,
$\lim_{x\rightarrow \infty}l(\eta x)/l(x)=1,\;\eta
>1$.
\end{definition}
The preceding class includes the well-known power law distributions,
e.g., $F(x)=1-1/x^\alpha, x\ge 1, \alpha>0$.

Let $\tilde{B}_i,1\le i \le m$  be the order statistics of
$B_{-i},1\le i\le m$ with the convention $\tilde{B}_0=\infty$ and
$\tilde{B}_{m+1}=0$. To make the notation uniform, we assume that
$\tilde{B}_0=\infty$, $\tilde{B}_{1}=0$ for $m=0$, and when it is
necessary to emphasize the total number of random variables, we
write explicitly $\tilde{B}_i^{(m)}\equiv \tilde{B}_i$.
\begin{definition}
Let $\mathcal{A}_k^{(m)}\triangleq \{\tilde{B}_{k+1}^{(m)}\le B_0<
\tilde{B}_k^{(m)}\}$ for $m\ge k\ge 0$.
\end{definition}

The asymptotic behavior of the sojourn time distribution for PS,
FBPS and SRPT has been extensively studied under heavy-tails,
e.g., see \cite{zwart98,NQ00,BBQ03,Predrag02,JEMO03} and the
references therein. We summarize these results for intermediately
regularly varying distributions in the following theorem, which
follows directly from our more general/refined result presented in
Theorem~\ref{theorem:main} in the following section. In order to
ease the notation we simply write $B\equiv B_0$ and $V\equiv V_0$.

For the rest of the paper, we assume that the system has reached
stationarity. Also, we use $H$ to denote a sufficiently large
positive constant. The value of $H$ is generally different in
different places,  for example, $H/2=H$, $H^{2}=H$, $H+1=H$, etc.
Furthermore, we use the following standard notation. For any two
real functions $a(t)$ and $b(t)$ and fixed $t_{0}\in
\real\cup\{\infty\}$ we will use $a(t)\thicksim b(t)$ as
$t\rightarrow t_{0}$ to denote $\lim_{t\rightarrow t_{0}}
[a(t)/b(t)]=1$. Similarly, we say that $a(t)\gtrsim b(t)$ as
$t\rightarrow t_{0}$ if $\lim\inf_{t\rightarrow t_{0}}a(t)/b(t)\ge
1$; $a(t)\lesssim b(t)$ has a complementary definition. In addition,
we say that $a(t)=o(b(t))$ as $t\rightarrow t_0$ if
$\lim_{t\rightarrow t_0} a(t)/b(t)= 0$. When $t_0=\infty$, we often
simply write $a(t)=o(b(t))$ without explicitly stating $t\rightarrow
\infty$ in order to simplify the notation.

\begin{theorem}\label{theorem:classic}
If $B\in {\cal IR}$ and $\expect B^\alpha<\infty$ for some
$\alpha>1$, then,
  under the PS, FBPS or SRPT discipline, we have, as $x\rightarrow \infty$,
\begin{equation*}
  \Pr \left[V>x \right] \sim \Pr \left[B>(1-\rho)x
  \right].
  \end{equation*}
\end{theorem}
The preceding asymptotic insensitivity of the sojourn (waiting) time
distribution on the scheduling discipline was first derived in
Theorems~5.2.3, 5.2.4 and 5.2.5 of \cite{NQ00} under somewhat more
restrictive conditions;
see also Theorems~2.2, 2.5 and 2.6 of \cite{BBQ03}. For PS, this
result was proved in Theorem~2.1 of \cite{JEMO03} using a novel
sample path approach that allows further extension of the result to
moderately heavy distributions, e.g., lognormal, see Theorem~3.1 of
\cite{JEMO03}. Furthermore, as noted in Appendix B of
\cite{Predrag02}, this sample path approach extends directly to SRPT
and FBPS scheduling disciplines. Our proof of
Theorem~\ref{theorem:main} in this paper relies directly on the
arguments developed by \cite{Predrag02,JEMO03}.

\subsection{Conditional Limits} The following theorem represents our first main
result, which implies Theorem \ref{theorem:classic} by
unconditioning on event $\mathcal{A}_k^{(m)}$, i.e., summing over
all $k$, $0\le k\le m$.
\begin{theorem}\label{theorem:main}
  If $B\in \mathcal{IR}$ and $\expect B^\alpha<\infty$ for some $\alpha>1$, then, under either PS or FBPS discipline,
   we have for fixed $k$, as $x\to \infty$,
\begin{equation}\label{eq:ps}
  \Pr \left[V>x, \mathcal{A}_k^{(m)} \right] \sim \Pr
  \left[B>\frac{(1-\rho)x}{(1+k)},
  \mathcal{A}_k^{(m)}\right] \sim  \frac{1}{k+1}{{m}\choose{k}} {  \Pr \left[B>\frac{(1-\rho)x}{k+1}
  \right]^{k+1} },
  \end{equation}
  and under the SRPT discipline,
  \begin{equation}\label{eq:SRPT}
  \Pr \left[V>x, \mathcal{A}_k^{(m)} \right] \sim  \Pr \left[B>
  (1-\rho)x,
 \mathcal{A}_k^{(m)} \right] \sim  \frac{1}{k+1}{{m}\choose{k}}  {  \Pr \left[B>(1-\rho)x
  \right]^{k+1} }.
  \end{equation}
\end{theorem}
\begin{remark}
These results can be easily extended to $GI/GI/1$ queue under the
FBPS discipline, and possibly under the SRPT as well using the
recent studies on SRPT by \cite{Nuy07}.
In order to provide a unified framework, we omit such possible
extensions here and restrict ourselves to the $M/G/1$ framework.
Furthermore, our focus in the second part of the paper is to exploit
this idea of relative job comparisons to design adaptive and
efficient approximation of SRPT, which we term \emph{comparison
scheduling}.
\end{remark}
\begin{remark}
Note that on $\mathcal{A}_k^{(m)}$, the distribution of $B$ has a
much lighter tail of the order of $\Pr \left[B>x \right]^{k+1}$ and,
thus, $\mathcal{A}_k^{(m)}$ partitions the probability space into
jobs of decreasing sizes as $k$ increases. Interestingly, the result
shows that, for the SRPT discipline, even the relatively much
smaller job receives the entire long-term residual capacity
$1-\rho$, while, for PS/FBPS, this smaller job shares equally the
residual capacity with the $k$ larger ones. Hence, SRPT outperforms
PS/FBPS for medium size jobs and therefore provides much better and
more uniform performance over a wide range of time scales, i.e., it
appears that SRPT generates extra capacity. Informally, we believe
that the explanation for this comes from the combined effect of the
SRPT prioritization mechanism and the fact that jobs of
``different'' sizes occur on different time scales. Hence, the
medium size jobs are basically not affected by the larger ones
because of the higher priority assigned to them and the larger jobs
are not impacted by the smaller ones due to the time scale
separation.
\end{remark}
In order to prove this theorem, we define the class of heavy-tailed
distributions ${\cal L}$ that contains subexponential distributions
and, in particular, the intermediately regularly varying class
${\cal IR}$, and establish the following two preliminary lemmas.
\begin{definition}\label{def:heavy}
A nonnegative random variable $X$ or its d.f. F is called
heavy-tailed $X\in {\cal L}$ (or $F\in {\cal L}$) if, for any
fixed $y\in \RR$,
$$\lim_{x\rightarrow \infty}\frac{\Pr[X>x-y]}{\Pr[X>y]}=1.$$
\end{definition}
\begin{lemma}
\label{lemma:orderstat2} Let $\{X_i \}_{0\leq i \leq m}$ be i.i.d.
random variables with $X_0 \in {\cal L}$ and,  denote the order
statistics of $X_1, X_2, \cdots, X_m$ by $\tilde{X}_1\geq
 \tilde{X}_2\geq  \cdots \geq \tilde{X}_m$ with $\tilde{X}_0=\infty$ and $\tilde{X}_{m+1}=0$, then, for any $m\ge k\ge 0$, as $x\rightarrow
 \infty$,
we have
\begin{align}\label{eq:two}
\Pr[X_0>x, \tilde{X}_{k+1}\le X_0< \tilde{X}_k]&\sim \Pr[X_0>x,
X_0< \tilde{X}_k]\nonumber \\&\sim \frac{1}{k+1}{m\choose
k}(\Pr[X_0>x])^{k+1}.
\end{align}
\end{lemma}
\begin{remark}
{\rm This result holds for all continuous distributions without
the assumption $X_0\in {\cal L}$. However, the assumption $X_0\in
{\cal L}$ is necessary in general since the result may not hold
for light-tailed lattice valued distributions. Here, easy
calculations show that the lemma does not hold for geometric
distribution $\Pr[X_i=j]=p^j(1-p), j\ge 0$, where we obtain for
$m=k=1$ and positive integer $x\in \nat$
\begin{align*}
\Pr[X_1> X_0> x]&=\expect\left[\ind\{X_0>
x\}p^{X_0+1}\right]=\frac{p}{1+p}(\Pr[X_0> x])^2.
\end{align*}}
\end{remark}
\begin{lemma}\label{lemma:compare}
 If two arrival processes
$\{(T_i,B_{1i}) \}_{i> -\infty}$\; and $\{(T_i,B_{2i}) \}_{i>
-\infty}$, satisfying $B_{1i} =0$ for $i< 0$, $B_{10}=B_{20} $, and
either $B_{1i}=B_{2i}$ or $B_{1i}=0$ for $i > 0$, are served with
SRPT discipline, then, the corresponding sojourn times $V_1$ and
$V_2$ for the customer arriving at $T_0$ satisfy $V_2\ge V_1$.
\end{lemma}
 The {\bf proofs} of Lemma
\ref{lemma:orderstat2} and \ref{lemma:compare} are presented in
the Appendix.

\noindent{\bf Proof of Theorem \ref{theorem:main}:} Label the
customer that arrives at time $T_0$, and define function
$R_0(t)\equiv R_{B_0}(t)$ for $t\ge 0$ to be the amount of remaining
work of the labeled customer at time $t$. Let $L_m$ be the number of
customers in the system just before time $T_{-m}$. For all the
customers arriving between $T_{-m}$ and $T_{0}$, define $B_{-i}^{0}$
to be the remaining service time of $B_{-i}, 1\le i\le m$ at time
$t=0$. For all the customers arriving before time $T_{-m}$, define
$B_{i}^{(e)}(m), 1 \le i \le L_{m}$ to be the remaining service time
at time $T_{-m}$ and $B_i^{(e)}(0)$ the remaining service time at
time $t=0$. Denote $x\wedge y\equiv \min(x,y)$.

\emph{1. Processor sharing discipline.} Similarly as in
\cite{JEMO03},
 we have the following min-plus identity
\begin{align}
\label{eq:identps}
 V_0=B_0+\sum_{i=1}^{m}B_{-i}^0\wedge
B_0+\sum_{i=1}^{L_{m}}B_i^{(e)}(0)\wedge
B_0+\sum_{i=1}^{N(V_0)}B_i\wedge R_0(T_i).
\end{align}

First, we establish an \textit{upper bound} for (\ref{eq:ps}).
Observing that the residual service $B_{-i}^0$ for customer $-i$ at
time 0 is upper bounded by its original job size and using
$B_i^{(e)}(0)\le B_i^{(e)}(m)$ as well as $R_0(T_i)\le B_0$, we
derive on set $\mathcal{A}_k^{(m)}$
\begin{align*}
V_0& \le B_0+\sum_{i=1}^{m}B_{-i}\wedge
B_0+\sum_{i=1}^{L_{m}}B_i^{(e)}(m)\wedge
B_0+\sum_{i=1}^{N(V_0)}B_i\wedge B_0\\  &\le
(k+1)B_0+(m-k)\tilde{B}^{(m)}_{k+1}+\sum_{i=1}^{L_{m}}B_i^{(e)}(m)\wedge
B_0+\sum_{i=1}^{N(V_0)}B_i\wedge B_0,
\end{align*}
where $\tilde{B}^{(m)}_{m+1}\equiv 0$ for $m=k$. Then, for
$0<\delta<1-\rho$, we have
\begin{align}\label{eq:psunion}
\Pr \left[V_0(1-\rho-\delta)>x,\mathcal{A}_k^{(m)} \right]&\le
\Pr\left[(k+1)B_0>(1-\delta)x,\mathcal{A}_k^{(m)} \right]
+\Pr\left[(m-k)\tilde{B}^{(m)}_{k+1}>\frac{\delta x}{3},
\mathcal{A}_k^{(m)} \right] \nonumber\\
&\quad+\Pr \left[W_{B\wedge
B_0}^{\rho+\delta }>\frac{\delta x}{3}, \mathcal{A}_k^{(m)}
\right]+\Pr \left[\sum_{i=1}^{L_{m}}B_i^{(e)}(m)\wedge
B_0>\frac{\delta x}{3},\mathcal{A}_k^{(m)} \right]\nonumber \\
&\eqdef I_1(x)+I_2(x)+I_3(x)+I_4(x),
\end{align}
where $W_{B\wedge B_0}^{\rho+\delta }$ denotes the stationary
workload in a queue with job sizes  $\{B_i\wedge B_0\}_{i\ge 1}$ and
service capacity $\rho+\delta$. Now, Lemma \ref{lemma:orderstat2}
implies
 \begin{align}\label{eq:psI1}
I_1(x) =\Pr\left[(k+1)B_0>(1-\delta)x,\mathcal{A}_k^{(m)}
\right]\sim \frac{1}{k+1}{m \choose
k}\left(\Pr\left[B_0>\frac{(1-\delta)x}{k+1} \right]\right )^{k+1}.
\end{align}
Then, denote the order statistics of $\left\{ B_{-i} \right\}_{0\leq
i\leq m}$ by $\{\tilde{B}^{(m+1)}_{i}\}_{0\leq i \leq m}$. For
$k=m$, we have $I_2(x) =0$. And, for $0\leq k \leq m-1$, we obtain,
from Lemma \ref{lemma:orderstat2} and $B_0 \in \mathcal {IR}$,
 \begin{align}\label{eq:I1}
I_2(x)&=\Pr\left[(m-k)\tilde{B}^{(m)}_{k+1}>\frac{\delta x}{3},
\mathcal{A}_k^{(m)} \right]\leq \Pr\left[\tilde{B}^{(m+1)}_{k+2}>\frac{\delta x}{3(m-k)}\right]\nonumber \\
&\sim {m+1 \choose k+2}\left(\Pr\left[B_0>\frac{\delta
x}{3(m-k)}\right]\right)^{k+2}=o(I_1(x)).
\end{align}

Following the same technique that was developed by \cite{JEMO03},
we have
\begin{align*}
I_3(x)&=\Pr \left[W_{B\wedge B_0}^{\rho+\delta }>\frac{\delta
x}{3}, \mathcal{A}_k^{(m)} \right]\nonumber \\&\le \Pr\left
[B_0>\delta ^2  x, \mathcal{A}_k^{(m)}
\right]\Pr\left[W_B^{\rho+\delta }>\frac{\delta
x}{3}\right]+\Pr\left[ W_{B\wedge \delta ^2 x}^{\rho+\delta
}>\frac{\delta x}{3}\right],
\end{align*}
which,  by Lemma 3.2 (i) in \cite{JEMO03}, implies that for $\delta$
small enough,
\begin{equation}\label{eq:psupI3}
I_3(x)=o\left(\Pr \left[B>\frac{x}{k+1} \right]^{k+1}
\right)=o(I_1(x)).
\end{equation}

Again, similarly as in \cite{JEMO03}, for any integer $n_0$, we
have
\begin{align}\label{eq:psupI44}
I_4(x)&=\Pr\left[\sum_{i=1}^{L_{m}}B_i^{(e)}(m)\wedge
B_0>\frac{\delta x}{3},\mathcal{A}_k^{(m)} \right]\nonumber \\
&=\sum_{n=1}^{\infty}(1-\rho)\rho^n \Pr\left[
\sum_{i=1}^{n}B_i^{(e)}(m) \wedge B_0 >\frac{\delta x}{3},
\mathcal{A}_k^{(m)}
\right]\nonumber \\
&\leq n_0 \Pr\left[ \sum_{i=1}^{n_0}B_i^{(e)}(m) \wedge B_0
>\frac{\delta x}{3}, \mathcal{A}_k^{(m)}
\right]  +  \sum_{n=n_0}^{\infty}(1-\rho)\rho^n \Pr\left[ B_0
> \frac{\delta x}{3n}, \mathcal{A}_k^{(m)}
\right]\nonumber \\
&\eqdef I_{41}+I_{42}.
\end{align}
Here, it is easy to see that
\begin{align}\label{eq:p1}
I_{41} &\leq n_0^2\Pr \left[B_1^{(e)}(m) > \frac{\delta x}{3n_0},
\mathcal{A}_k^{(m)}\right] \Pr \left[B_0>\frac{\delta x}{3n_0},
\mathcal{A}_k^{(m)} \right]\nonumber \\
&=o\left( \Pr \left[B_0> x, \mathcal{A}_k^{(m)} \right] \right).
\end{align}
Furthermore, since
\begin{align}
s \eqdef \sup_{x\in [0,\infty)} \frac{\Pr\left[B_0> x,
\mathcal{A}_k^{(m)}\right]}{\Pr\left[B_0>2x, \mathcal{A}_k^{(m)}
\right]}<\infty, \nonumber
\end{align}
we obtain that, for any $\epsilon>0, n\ge 1$, there exists
$K_\epsilon>0$ such that
\begin{align}
\Pr\left[ B_0 > \frac{\delta x}{3n}, \mathcal{A}_k^{(m)} \right]
&\leq s^{\lceil \log_2(n) \rceil} \Pr\left[B_0 > \frac{\delta x}{3},
\mathcal{A}_k^{(m)}\right]
 \leq K_\epsilon(1+\epsilon)^n \Pr\left[B_0 > \frac{\delta x}{3}, \mathcal{A}_k^{(m)}\right],
 \nonumber
\end{align}
which, by choosing $\epsilon$ small enough with $\eta \eqdef \rho
(1+\epsilon)<1$, yields
\begin{align}\label{eq:p2}
\sum_{n=n_0}^{\infty}(1-\rho)\rho^n \Pr\left[ B_0 > \frac{\delta
x}{3n}, \mathcal{A}_k^{(m)} \right] &\leq
\sum_{n=n_0}^{\infty}(1-\rho)\rho^n K_{\epsilon}(1+\epsilon)^n
\Pr\left[B_0
> \frac{\delta x}{3},
\mathcal{A}_k^{(m)}\right] \nonumber \\
&\leq \frac{(1-\rho)K_{\epsilon} \eta^{n_0}}{1-\eta} \Pr\left[B_0 >
\frac{\delta x}{3}, \mathcal{A}_k^{(m)}\right].
\end{align}
By combining (\ref{eq:psupI44}), (\ref{eq:p1}) and (\ref{eq:p2}),
and then passing $n_0\to \infty$, we obtain
 \begin{equation}
I_4(x)=o\left(\Pr\left[B>\frac{x}{k+1}\right]^{k+1}\right)=o(I_1(x)),
\nonumber
\end{equation}
which, in conjunction with (\ref{eq:psunion}), (\ref{eq:psI1}),
(\ref{eq:I1}),  (\ref{eq:psupI3}), and by passing $\delta \to 0$,
yields
\begin{equation}\label{eq:asympU}
\Pr \left[V_0>x,\mathcal{A}_k^{(m)} \right]\lesssim \Pr
\left[B_0>\frac{(1-\rho)x}{k+1},\mathcal{A}_k^{(m)} \right].
\end{equation}

 Next, we prove a \textit{lower bound} for
(\ref{eq:ps}). Observe that within $\mathcal{A}_k^{(m)}$, we have
\begin{align}\label{eq:lowerV}
V_0 \geq B_0+\sum_{i=1}^{m}B_{-i}^0\wedge B_0
+\sum_{i=1}^{N(V_0)}B_i\wedge R_0(T_i)\ge
(k+1)B_0+mT_{-m}+\sum_{i=1}^{N(V_0)}B_i\wedge R_0(T_i),
\end{align}
where in the last inequality we applied $(x-y)\wedge z\ge x\wedge
z-y$ for any $x,y,z\geq 0$; recall that $T_{-m}<0$. Then, using the
same arguments as in equation (3.11) in the proof of Theorem~2.1 in
\cite{JEMO03},  and the properties of $\mathcal{A}_k^{(m)}$, for
$B_0\in {\cal IR}$, we have
\begin{equation}\label{eq:asympL}
\Pr \left[V_0(1-\rho)>x,\mathcal{A}_k^{(m)} \right] \gtrsim  \Pr
\left[B_0>\frac{x}{k+1},\mathcal{A}_k^{(m)} \right].
\end{equation}
Combining (\ref{eq:asympU}) and (\ref{eq:asympL}) completes the
proof of (\ref{eq:ps}) for PS.

 \emph{2. FBPS discipline.} The proof is based on the sojourn time identity for FBPS
\begin{align*}
V_0&=B_0+W_{B\wedge B_0}(T_0)+\sum_{i=1}^{N(V_0)}B_i\wedge B_0,
\end{align*}
where $W_{B\wedge B_0}(T_n)$ denotes the stationary workload at
$T_n$ in a queue with Poisson arrival job sizes equal to
$\{B_i\wedge B_0\}_{-\infty<i<n}$ and capacity $1$; recall that
$T_0=0$.

First, we establish an \textit{upper bound}. Observe that within the
set $\mathcal{A}_k^{(m)}$,
\begin{align}
V_0\le (k+1)B_0+(m-k)\tilde {B}_{k+1}^{(m)}+W_{B\wedge
B_0}(T_{-m})+\sum_{i=1}^{N(V_0)}B_i\wedge B_0,\nonumber
\end{align}
which,  for $0<\delta<1-\rho$, implies
\begin{align}\label{eq:fbpsunion}
\Pr \left[V_0(1-\rho-\delta)>x,\mathcal{A}_k^{(m)} \right]& \le
\Pr\left[(k+1)B_0>(1-\delta)x,\mathcal{A}_k^{(m)} \right]
  +\Pr\left[(m-k)\tilde{B}_{k+1}^{(m)}>\frac{\delta x}{3}, \mathcal{A}_k^{(m)} \right]\nonumber\\
&\;\;\; +\Pr \left[W_{B\wedge B_0}(T_{-m})>\frac{\delta
x}{3},\mathcal{A}_k^{(m)}\right] +\Pr \left[W_{B\wedge
B_0}^{\rho+\delta
}>\frac{\delta x}{3},\mathcal{A}_k^{(m)} \right]\nonumber \\
&\eqdef I_1(x)+I_2(x)+I_3(x)+I_4(x).
\end{align}
Using the same arguments as in the proof of the upper bound for
the PS case,  we obtain
\begin{equation}\label{eq:fbpsI1}
I_1(x)\sim \frac{1}{k+1}{m \choose k
}\left(\Pr\left[B_0>\frac{(1-\delta)x}{k+1} \right]\right )^{k+1},
\end{equation}
and similarly as in (\ref{eq:I1}), (\ref{eq:psupI3}), it follows
that $I_2(x)=o(I_1(x)), I_3(x)=o(I_1(x)), I_4(x)=o(I_1(x))$.
Therefore, by (\ref{eq:fbpsunion}) and (\ref{eq:fbpsI1}), we have
\begin{equation}\label{eq:fbpsupper}
\Pr \left[V_0>x,\mathcal{A}_k^{(m)} \right]\lesssim \Pr
\left[B_0>\frac{(1-\rho)x}{k+1},\mathcal{A}_k^{(m)} \right].
\end{equation}

For a \textit{lower bound}, within $\mathcal {A}_k^{(m)}$, we
obtain
\begin{align*}
V_0\ge (k+1)B_0+T_{-m}+\sum_{i=1}^{N(V_0)}B_i\wedge B_0,
\end{align*}
which is further lower bounded by the righthand side of
(\ref{eq:lowerV}). Combining (\ref{eq:asympL}) and
(\ref{eq:fbpsupper}) completes the proof of (\ref{eq:ps}) for FBPS.

 \emph{3. SRPT discipline.}  A similar sojourn time
identity as in (\ref{eq:identps}) can be derived for SRPT,
\begin{align*}
\label{eq:identsrpt} V_0=B_0+\sum_{i=1}^{L_{m}}B_i^{(e)}(0)\ind
\{B_i^{(e)}(0)\le B_0\} +\sum_{i=1}^m B_{-i}^0\ind\{B_{-i}^0\le
B_0\}+\sum_{i=1}^{N(V_0)}B_i\ind \{B_i< R_0(T_i)\},
\end{align*}
where we use the  convention that the jobs with earlier arrivals are
served first in the case of equal remaining service times.

 First, we prove a \textit{lower bound}. For $l>0$,
define $B_{li} = 0$ for $i< 0$, and  $B_{l0} = B_0$, $B_{li} = B_i
\ind(B_i \leq l)$ for $i
> 0$. For the new queueing system with the arrival process $\{(T_i,B_{li})
\}$, denote by $\{W_{l}(t)\}_{t\geq 0}$  the workload in the system
without the labeled customer. Now, define the stopping time $T_{l0}
\eqdef \inf \{ t: R_0(t) \leq l \}$ and the corresponding residual
capacity without the labeled customer
$C(t)=\int_0^t\ind(W_l(s)=0)ds.$ Clearly,
\begin{equation}\label{eq:srptC}
  \expect[C(t)]\sim \left(1- \rho_l \right) t \;\; \text{ as $t\rightarrow
\infty$},
\end{equation}
where $\rho_l=\lambda \expect[B\ind(B\le l)]=\lim_{t\rightarrow
\infty} \Pr[W_l(t)>0]$.
 When $B_0>l$, all the arrivals after time
$T_0=0$ have shorter job requirements than the remaining service
time of the labeled customer before time $T_{l0}$, and thus, the
labeled customer can only receive service when there are no other
customers present in the queue except itself. Therefore, conditional
on $\{B_0>l\}$, we have
\begin{equation}
\label{eq:t0}
 C(T_{l0})=B_0-l.
\end{equation}
Next, by the standard queueing stability result and
(\ref{eq:srptC}), we have, for $\epsilon>0$,
$$Z\eqdef\sup_{t\ge 0}\left
(C(t)-(1-\rho_l+\epsilon)t\right)<\infty.$$
 From (\ref{eq:t0}),
$V_{l0} \geq T_{l0}$ and the monotonicity of $C(t)$, we obtain,
conditional on $\{B_0>l\}$,
\begin{align*}
Z\geq C(V_{l0})-(1-\rho_l+\epsilon)V_{l0}\geq
B_0-l-(1-\rho_l+\epsilon)V_{l0},
\end{align*}
which, for large $x$, implies
\begin{align}\label{eq:srptLower2}
\Pr\left[V_{l0}>x,\mathcal{A}_k^{(m)}\right]&\geq \Pr\left[B_0>l,
B_0-l-Z
> (1-\rho_{l}+\epsilon)x,\mathcal{A}_k^{(m)}\right]\nonumber\\
&\geq
\Pr\left[B_0-l>(1+\epsilon)(1-\rho_l+\epsilon)x,\mathcal{A}_k^{(m)}\right]
-\Pr\left[Z>\epsilon (1-\rho_l+\epsilon)x\right].
\end{align}
 Furthermore,
since the service requirements $\{B_{li}\}_{i\geq 1}$ are bounded by
$l$, the busy period distribution of the corresponding workload
$W_l(t)$ is exponentially bounded (e.g., see \cite{NZ06,Pal06}),
implying that there exists $\delta>0$, such that
$\Pr[Z>x]=O(e^{-\delta x})$. This bound and (\ref{eq:srptLower2}),
combined with Lemma \ref{lemma:compare} and $B\in {\cal IR}$, yield
\begin{align*}
 \lim_{x\rightarrow \infty} \frac{\Pr\left[ V_0 > x,\mathcal{A}_k^{(m)}\right]}
    {\Pr\left[B_0>(1-\rho)x,\mathcal{A}_k^{(m)}\right]}
   \geq \lim_{x\rightarrow
\infty}\frac{\Pr\left[B_0>(1+\epsilon)(1-\rho_l+
\epsilon)x,\mathcal{A}_k^{(m)}\right]}{\Pr\left[B_0>(1-\rho)x,\mathcal{A}_k^{(m)}\right]}.
\end{align*}
Passing $l\to \infty$, $\epsilon \to 0$ in the preceding inequality,
we obtain the lower bound for SRPT.

For an \emph{upper bound}, since the number of customers in system
for SRPT at any time is not larger than the number of customers in
system for any other rule applied on the same sequence of arrivals
and service requirements, as shown by \cite{SCH68}, we use the
stationary number of customers $L_m^{(PS)}$ at time $T_{-m}$ in the
corresponding PS queue to upper bound $L_m$. Furthermore, the
workload $W$ observed at time $T_{-m}$ is an upper bound for the
residual work $R_i$ of a customer at time $T_{-m}$. Therefore,
\begin{align*}
 V_0\le B_0 +\sum_{i=1}^{L_m^{(PS)}}W\wedge B_0+\sum_{i=1}^m
B_{-i}\ind \{B_{-i}\leq
B_{0}-T_{-m}\}+\sum_{i=1}^{N(V_0)}B_i\wedge B_0,
\end{align*}
which, for any $0<\delta<1-\rho$,  yields
\begin{align}
\label{eq:srptunion}
 \Pr\left[V_0(1-\rho-\delta)>x,\mathcal{A}_k^{(m)}
\right]& \leq \Pr\left[B_0>(1-\delta)x,\mathcal{A}_k^{(m)}
\right]\nonumber\\
&\;\;\;\; +m\Pr\left[B_{-1}\ind \{B_{-1}\leq
B_{0}-T_{-m}\}>\frac{\delta
x}{3m},\mathcal{A}_k^{(m)}\right] \nonumber \\
 &\;\;\;\; +
\Pr\left[ \sum_{i=1}^{L_m^{(PS)}}W\wedge B_0
>\frac{\delta x}{3m},
            \mathcal{A}_k^{(m)} \right]+\Pr \left[W_{B\wedge
B_0}^{\rho+\delta
}>\frac{\delta x}{3}, \mathcal{A}_k^{(m)} \right]\nonumber \\
 & \eqdef I_1(x)+I_2(x)+I_3(x)+I_4(x).
\end{align}

Similarly as in the proof of the upper bound for PS, we have
\begin{equation}\label{eq:srptUpperI1}
I_1(x)\sim  \frac{1}{k+1}{m\choose
k}\left(\Pr\left[B_0>(1-\delta)x\right ]\right )^{k+1}
\end{equation}
and
\begin{equation}\label{eq:srptUpper2}
I_3(x)=o(I_1(x)),\;\; I_4(x)=o(I_1(x)).
\end{equation}

 The only
difference, as compared to the PS case, is to evaluate $I_2(x)$.
Noting that $\mathcal{A}_k^{(m)}$ is a subset of
$$\{ B_0 \leq \tilde{B}_k \}=\bigcup_{1\leq
i_1<\cdots<i_k\leq m}\{ B_{-i_1} \geq B_0, \cdots, B_{-i_k}\geq
B_0 \},$$ we obtain
\begin{align}\label{eq:srptUpper1}
\frac{I_2(x)}{m}&\leq \Pr\left[B_{-1}>\frac{\delta x}{3m},
B_{0}\leq \tilde
  {B}_{k}, B_{-1}<
B_0\right]+\Pr\left[ B_{0}+ \mid T_{-m} \mid \geq B_{-1}
>\frac{\delta x}{3m}, B_{0}\leq \tilde
  {B}_{k}, B_{-1}\geq B_0 \right]\nonumber\\
& \eqdef P_{1}+P_{2},
\end{align}
where $P_1$ is derived by upper bounding the indicator function in
$I_2(x)$ by 1. To estimate $P_1$,  we use
\begin{align}\label{eq:i1}
 P_{1}& \leq {m-1 \choose k} \Pr\left[ B_{-1}> \frac{\delta x}{3m}, B_0 > B_{-1},
            \bigcap_{2\leq i \leq k+1} \{ B_{-i} \geq B_0 \}
  \right]\nonumber\\
  &\leq {m-1 \choose k} \left( \Pr\left[ B_{-1} > \frac{\delta x}{3m}\right]
  \right)^{k+2}=o\left( I_1(x)\right).
\end{align}
Next, for $  y \eqdef \delta x/(3m)$,  it is easy to see
\begin{align}
P_{2}\leq &{m-1 \choose k-1}\Pr \left[ B_0+ \mid T_{-m} \mid \geq
B_{-1}
> y,
         \bigcap_{1\leq i \leq k} \{ B_{-i} \geq B_0 \} \right], \nonumber
\end{align}
where the preceding probability is further bounded by
\begin{align}\label{eq:sqptUpper1}
&\Pr\left[ B_{-1}>y, B_0 \leq B_{-1} \leq B_0+\sqrt{y}\right]
   \Pr\left[ B_0 \geq y- \sqrt{y}\right]^{k-1}+ \Pr[\mid T_{-m} \mid >\sqrt{y}]\nonumber\\
 &\leq \left(  \Pr[ B_{-1}\geq y, B_{-1}\leq B_0+\sqrt{y} ]-\Pr[B_{-1}\geq y, B_{-1}<B_0 ]
\right)\nonumber \\
&\quad  \times \Pr\left[ B_0 \geq y- \sqrt{y}\right]^{k-1}
+me^{-\lambda \sqrt{y}/m}.
\end{align}
Since $B_{0}, B_{-1} \in \mathcal {IR} $ and $\Pr[ B_{-1}\geq y,
B_{-1}\leq B_0+\sqrt{y} ]\lesssim  \Pr\left[ B_0 \geq y
\right]^{2}\sim \Pr[B_{-1}\geq y, B_{-1}<B_0 ]$, the right-hand
side of inequality (\ref{eq:sqptUpper1}) is asymptotically equal
to
$$o\left( \Pr\left[ B_0 \geq y \right]^{k+1}
\right)=o(I_1(x)),$$ which, in conjunction with
 (\ref{eq:i1}) and (\ref{eq:srptUpper1}),
implies $I_2(x)=o(I_1(x))$. Finally, by replacing
(\ref{eq:srptUpperI1}), (\ref{eq:srptUpper2}) and the preceding
estimation of $I_2(x)$
 in
(\ref{eq:srptunion}), and then passing $\delta \to 0$, we finish the
proof.
 \qed

\section{Adaptive and Scalable Comparison Scheduling}
Motivated by our conditional limits presented in Section
\ref{s:limit}, we propose a novel adaptive and scalable comparison
scheduling scheme.
\subsection{Comparison Splitting} \label{sec:comparison}
In this section, we describe a new adaptive job classification
mechanism that we term {\it comparison splitting}. The
classification is based on relative size comparison of the arriving
job to the previous $m$ arrivals, $m\ge 1$. Specifically, if an
arriving job is smaller than $k$ and larger than $m-k$ of the
previous $m$ jobs, it is routed into class $k$, $0\leq k \leq m$.

More formally,  upon the arrival of job $i\ge 0$, we define
$\tilde{B}_{i1}\ge \tilde{B}_{i2}\ge\cdots\tilde{B}_{im}$ to be the
order statistics of $\{B_{i-m}, B_{i-m+1}$, $\cdots$, $ B_{i-1}\}$
with $\tilde{B}_{i0}=\infty$ and $\tilde{B}_{i(m+1)}=0$. Then, if
$\tilde{B}_{i(k+1)}\le B_i< \tilde{B}_{ik}$, the new arrival $B_i$
is routed to class $k, 0\le k\le m$ and the $i$th arrival in class
$k$ is denoted as $B_i^{(k)}$.
In order to initiate the comparison
splitting process, assume that $B_i, -m\leq i\leq -1$ are already
known; otherwise, one can simply set ${B}_{i}\equiv 0, -m\leq i\leq
-1$.

Here, we exemplify our splitting mechanism for $m=3$ by dividing
jobs into four classes S (small), M (medium), L (large) and XL
(extra large) with the following rule,
\begin{displaymath}
B_i\in\left\{\begin{array}{lll} S & \textrm{if $B_i< \tilde{B}_{i3}$},\\
M &\textrm{if $\tilde{B}_{i3}\le B_i< \tilde{B}_{i2}$},\\
L &\textrm{if $\tilde{B}_{i2}\le B_i< \tilde{B}_{i1}$},\\
XL &\textrm{if $\tilde{B}_{i1}\le B_i;$}\\
\end{array}\right.
\end{displaymath}
this example is depicted in Figure~\ref{fig:splitter} (A).

Now, we argue that our comparison splitting actually does order jobs
into classes that contain smaller jobs for larger class indexes.
 Indeed,  when $B\in {\cal L} $, Lemma~\ref{lemma:orderstat2} yields
\begin{equation}
\label{eq:lo} \Pr\left[B_1^{(k)}>x\right]\sim\frac{1}{k+1}{m\choose
k}\Pr[B>x]^{k+1} \;\; \text{ as $x\rightarrow \infty$},
\end{equation}
which implies a decreasing distribution tail when $k$ increases.
Since the preceding expression is only an asymptotic result, it
does not provide information on the possible ordering of the
distributions $\Pr\left[B_1^{(k)}>x \right]$ for finite $x$. We
address this question in the following example.

\begin{example}
  In this example we simulate the performance of the comparison
 splitter for $m=3$ ($4$ classes).
  Assume that the job sizes are distributed as
power law $F(x)=1-1/x^\alpha$ with $\alpha=1.44$,  which is the
empirically measured file distribution by \cite{JEM00b}; see Figure
1 on p.$~577$ therein. For a sample of $10^7$ trials, we plot the
simulated distributions of jobs for each class in
Figure~\ref{fig:splitter} (B). From the figure, it can be observed
that the distributions $\Pr\left[B_1^{(k)}>x \right]$ are properly
ordered for all values of $x$ and $k$, not only for the asymptotic
ones.

\begin{figure}\label{fig:splitter}
\begin{tabular}{m{7cm} m{7.5cm}}
  \centering \epsfig{file=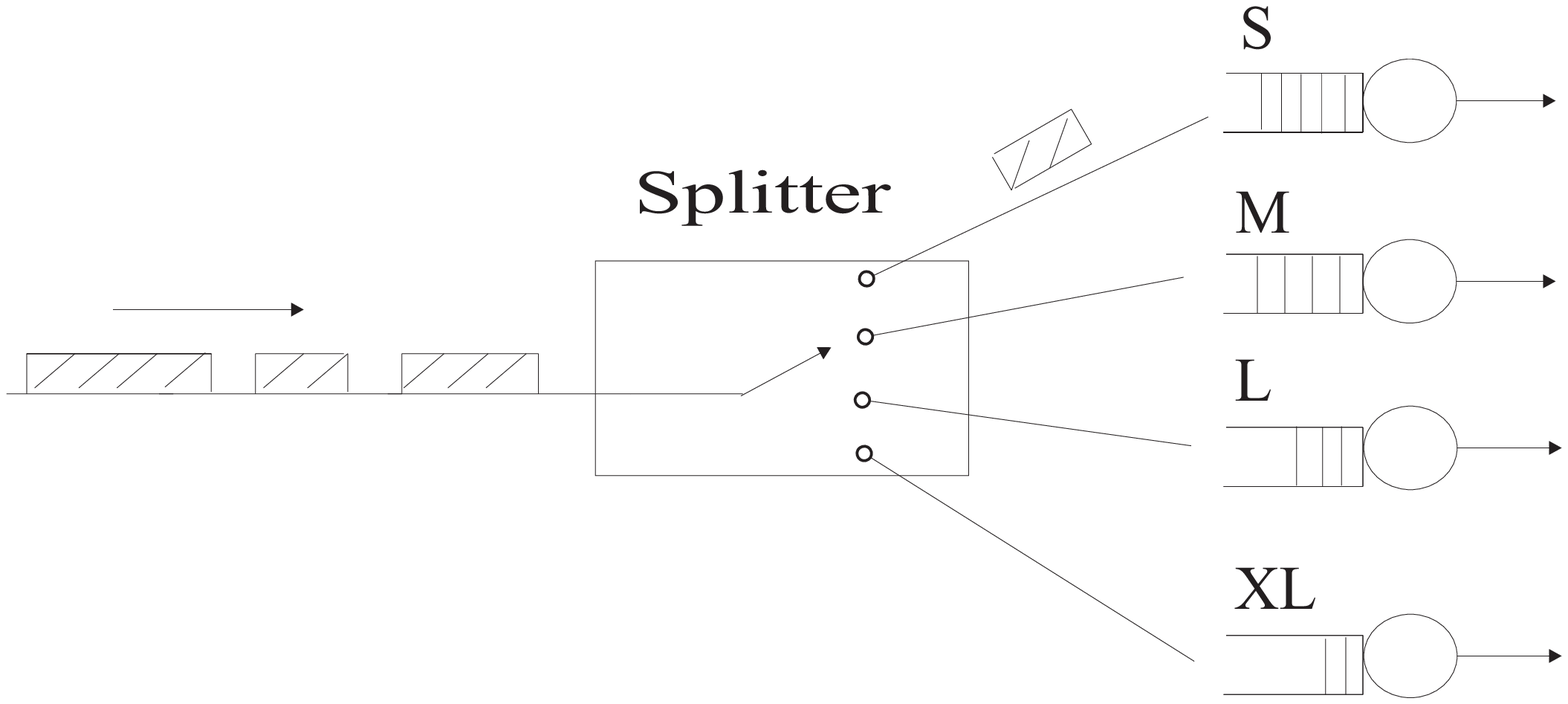,width=7cm,height = 3.5cm}
& \epsfig{file=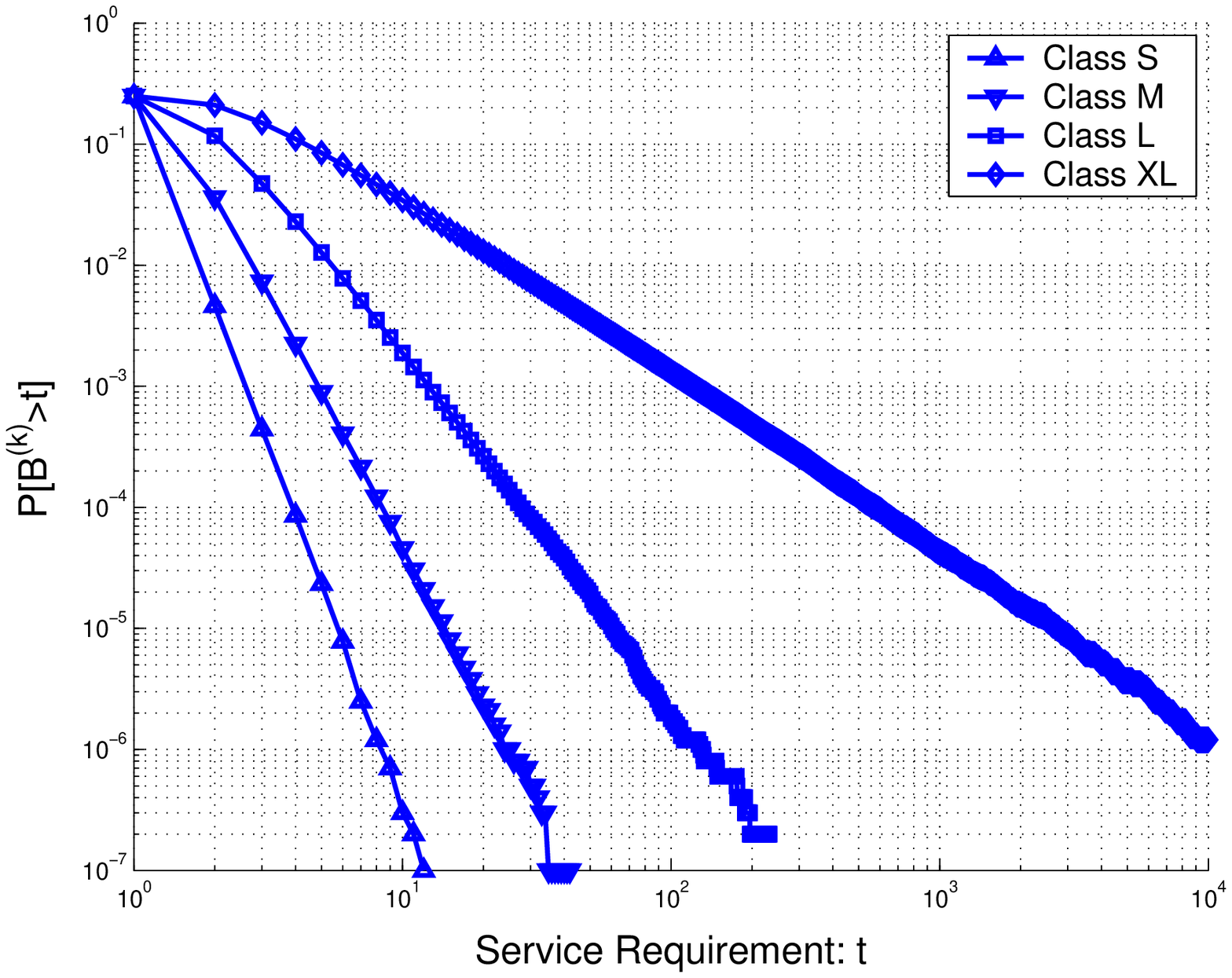,width=7.9cm,height = 5.5cm}\\
\centering (A) Comparison splitter.& \centering (B) Job size
distributions.
\end{tabular}
\caption{A comparison splitter with $m=3$ and job size
distributions for four different classes.}
\end{figure}
\end{example}

Based on the previous analysis and simulation example, we can see
that our comparison splitter has the following advantages:
\begin{itemize}
\item it is adaptive since the comparing thresholds are defined by
the preceding $m$ arrivals; \item it is scalable  because the system
only needs to know the sizes of the previous $m$ jobs; 
\item it provides accurate job classification as shown by equation
(\ref{eq:lo}) and Figure~\ref{fig:splitter} (B).
\end{itemize}

Although our comparison splitter is very likely to provide a
satisfactory ordering of distributions
$\Pr\left[B_1^{(k)}>x\right]$, it may make errors on a sample path
basis. Namely, it can occasionally classify smaller jobs into
 classes of smaller indexes and vice versa, and thus, possibly give a less
accurate classification than a splitting mechanism that uses fixed
thresholds. However, this possible small loss of accuracy is a
fundamental tradeoff to gain the adaptability that is highly
desirable in practice.

\subsubsection{Refined Splitting} \label{sec:ref} From the
description of the comparison splitter, we can see that its adaptive
thresholds are determined by the order statistics of the previous
$m$ arrivals. Thus, it is reasonable to expect that, at least for a
stationary input, the accuracy of the classification will increase
if we obtain these thresholds using a longer history (than the
preceding $m$ arrivals). However, the increase of history may reduce
the adaptability and add to the complexity of the algorithm.

Here, we describe one such improved comparison splitter that is
based on the order statistics of the preceding $ml, l\ge 1$ arrivals
and parameterized by $(m,l)$. Among other reasons, we continue to
use the order statistics since the ordered list is easy to maintain
dynamically. The splitter works as follows. At the time of arrival
of a new job $i$, the algorithm maintains the job sizes of the
previous $ml$ arrivals, and orders them as $\tilde{B}_{i1}\ge
\tilde{B}_{i2}\ge \cdots \ge \tilde{B}_{(i,ml)}$; when needed, we
use the notation $\tilde{B}_{(i,j)}\equiv \tilde{B}_{ij}$ for
improved clarity. We pick the subsequence $\tilde{B}_{(i,jl)}, 1\leq
j \leq m$ as the thresholds with
$\tilde{B}_{(i,0)}=\infty,\tilde{B}_{(i,(m+1)l)}=0$, and then,  the
new arrival is grouped into class $k$ if its size lies in $[
\tilde{B}_{(i,kl)}, \tilde{B}_{(i,(k-1)l)})$, $1\leq k \leq m+1$.

In terms of engineering applications, this refined splitting
algorithm is appealing because it can improve the accuracy for other
types of arrivals, such as  dependent processes and concentrated
discrete distributions of job sizes. In order to measure how well
the refined splitter classifies the input sequence, we compare the
output of the refined splitter with a perfectly ordered input
sequence. Denote the input sequence by $\{B_i\}_{1\leq i \leq n}$,
the output of the refined splitter by $\{O_i\}_{1\leq i \leq n}$,
and the increasing order of $\{B_i\}$ by $\{S_i\}_{1\leq i \leq n}$.
The output of the refined splitter $\{O_i\}$ is obtained by
concatenating sequentially class $j+1$ after class $j$ for all $1
\leq j\leq m-1$. Now, define the error rate to be
 $$\eta(n) \eqdef \frac{1}{n}\sum_{i=1}^{n}
 \ind\{O_i\neq S_i\}.$$
 \begin{lemma}
 \label{lemma:refined}
  For any fixed $0<\epsilon<1$, fixed $m$ large
enough, and an i.i.d. input sequence $\{B_i\}_{1\leq i\leq n}$
taking finite number of values $\Pr[B_1=b_j]=p_j, 1\leq j \leq v$
with the splitter initialized by $ml$ i.i.d. random copies of $B_1$
that are independent of $\{B_i\}_{1\leq i\leq n}$, there exists
$H_{\epsilon}, \xi_{\epsilon}>0$, such that
\begin{equation*}\label{eq:refined}
\Pr[\eta(n)> \epsilon] \leq  H_{\epsilon }e^{-\xi_{\epsilon}
\min(l,n)}.
\end{equation*}
\end{lemma}
The {\bf proof} of Lemma~\ref{lemma:refined} is presented in the
appendix.

\subsection{Queueing Analysis} \label{sec:qana} In this section, we study the queueing
performance of our comparison based scheduler assuming that jobs
arrive according to a stationary renewal process $\{T_n\}$,
$T_{-1}<0\le T_0$ with finite mean $\expect[T]<\infty$, where $T\eqd
T_1-T_{0}$. The job sizes $\{B_n\}$ before the splitting are i.i.d
and independent of $\{T_n\}$. To simplify the notation and analysis
in this section, we say that the $i$th arrival to class $k$ is equal
to $B_i^{(k)}=B_i\ind\{\tilde{B}_{i(k+1)}\le B_i< \tilde{B}_{ik}\}$.
This notation takes into account all the original arrival points
even if $B_i\ind\{\tilde{B}_{i(k+1)}\le B_i< \tilde{B}_{ik}\}=0$.
The addition of zero size jobs in each class has no impact on
queueing, but simplifies the exposition.

In Theorem \ref{theorem:isolate}, we characterize the workload
asymptotics when each class is served in isolation. Then, in
Theorem \ref{theorem:static}, we study the workload asymptotics of
each individual class assuming that all the classes are served
jointly according to a static priority discipline.

\subsubsection{Queueing in Isolation}
\label{ss:qi}We first study the queueing characteristics of each
class $k$ when it is served in isolation with capacity $c_k$,
$0\le k\le m$. We use $W^{(k)}$ to denote the stationary workload
of class $k$ and define $B^{(k)}\eqd B_1^{(k)}$.

\begin{theorem}\label{theorem:isolate}
If $\Pr[B>x]=l(x)/x^{\alpha} \in {\cal R}_{\alpha},\, \alpha > 1$
and $\expect[B^{(k)}] <c_k\expect [T]$, then, as $x\to \infty$,
\begin{align*}
  \Pr\left[ W^{(k)} >x  \right]&\sim \frac{1}{c_k \expect
  [T]-\expect[
  B^{(k)}]}\int_{x}^{\infty} \Pr \left[B^{(k)} > u
  \right]du \\
  &\sim \frac{1}{(k+1)(c_k \expect [T]-
  \expect[B^{(k)}])}{m\choose k}\frac{l(x)^{k+1}}{x^{\alpha k+\alpha -1}}.
\end{align*}
\end{theorem}
\begin{remark}
{\rm Note that this theorem indicates that the workload distribution
decays faster for larger $k$}. To be more specific, the tail of the
workload distribution for class $k$ decays as $x(\Pr[B>x])^{k+1}$
and, thus, the jobs will have the waiting time distribution of the
same order if served under FIFO. If, for example, each class were
served according to PS/FBPS, one can expect that the waiting times
will be of the same order as $(\Pr[B>x])^{k+1}$,  as in our Theorem
\ref{theorem:main}. However, this is much more difficult to prove
because of the dependency in $\{B_n^{(k)}\}$.
\end{remark}
\begin{remark}
{\rm Note that the result of Theorem \ref{theorem:isolate} is of the
same form as the one derived by \cite{PAK75} for the $GI/GI/1$
queue. However, Pakes's result does not apply directly to our case
since $\{B_n^{(k)}\}$ is $m$-dependent. For generalizations of
Pakes's result to dependent processes see \cite{JLA95g,ASS99}. Note
that, in principle, the approach from \cite{ASS99} can be applied to
prove our theorem. Instead, we present a direct proof that may be of
independent interest.}
\end{remark}
In order to prove this theorem, we need the following definitions
and lemmas. Define the partial sum of a stationary process
$\{X_n\}_{n\in
   \nat}$, where $X_n\in {\cal R}_\alpha$, as follows, $S_0=0$,
   \begin{equation}\label{eq:pts}S_n = \sum_{i=1}^{n}X_i,\;\; n\ge 1.\end{equation}
\begin{definition}
For a stationary process $\{X_n\}_{n\in \nat}$ and $m\in \nat$, we
say the process is $m$-dependent if $X_n$ is independent of
$\{X_i\}_{i< n-m}$ for
   all $n$.
\end{definition}

\begin{lemma}\label{eq:psum}
  If we define $$S_n^{(m)} \eqdef \sum_{i=0}^{\lfloor \frac{n}{m} \rfloor}X_{im+1},$$ then
  \begin{equation*}
   \Pr \left[\sup_{n\ge 0}S_n>x \right] \leq m \Pr \left[\sup_{n\ge 0}S_n^{(m)}> \frac{x}{m} \right].
  \end{equation*}
\end{lemma}
\begin{proof}
Define $$S_n^{(m,j)} = \sum_{i=0}^{\lfloor \frac{n}{m}
\rfloor}X_{im+j},$$ where $1\leq j\leq m$, and observe that $S_n
\le \sum_{j=1}^{m} S_n^{(m,j)}.$ Therefore,
\begin{align*}
\Pr \left[\sup_{n\ge 0}S_n>x \right] &=  \Pr \left[\sup_{n\ge 0}
\sum_{j=1}^{m} S_n^{(m,j)}>x \right]
  \leq \Pr \left[ \sum_{j=1}^{m} \sup_{n\ge 0}S_n^{(m,j)}>x \right] \nonumber \\
   &\leq \sum_{j=1}^m  \Pr \left[ \sup_{n\ge 0}S_n^{(m,j)}> \frac{x}{m}
   \right]\leq m  \Pr \left[ \sup_{n\ge 0}S_n^{(m)}> \frac{x}{m}
   \right],
\end{align*}
where the last equality follows from the stationarity of
$\{X_n\}$.\qed
\end{proof}

\begin{lemma}\label{lemma:geqHx}
   For a stationary $m$-dependent process $\{X_n\}_{n\in \nat}$ with
   mean
   $\expect X_1=-\delta<0$ and $X_1\in {\cal R}_{\alpha}$, we have
  \begin{equation*}
   \Pr \left[\sup_{n\geq Hx}S_{n}> x \right] \leq \frac{1}{H^{\alpha-1}}
   O\left(\int_x^\infty\Pr[X_1>u]du\right) .
  \end{equation*}
\end{lemma}

\begin{proof} For simplicity of notation, in this section, we assume that
$Hx\in \nat$. Then, we define $M\eqdef \sup_{n\geq 0}S_{n}$ with
$\expect[X_n]=-\delta$, and note that
 \begin{equation*}
 \sup_{n\geq Hx}S_{n}=
S_{Hx}+\sup_{n\geq Hx}(S_{n}-S_{Hx}).
\end{equation*}
Since the process $\{X_n\}$ is stationary, we obtain
$$\sup_{n\geq
Hx}(S_{n}-S_{Hx})\eqd M,$$ and therefore, $\Pr \left[ \sup_{n\geq
Hx}S_{n}>x \right]$ is upper bounded by
\begin{align*}
\Pr \left[S_{Hx}+\frac{\delta Hx}{2}+\sup_{n\geq Hx}(S_{n}-S_{Hx})
-\frac{\delta Hx}{2}> 0 \right] & \leq \Pr\left[S_{Hx}+
\frac{3\delta Hx}{4}>\frac{\delta Hx}{4} \right]+ \Pr
\left[M>\frac{\delta
Hx}{2} \right]\nonumber\\
& \eqdef I_1+I_2.
\end{align*}
{From} the result of \cite{PAK75} and Lemma \ref{eq:psum},
recalling that $X_1 \in {\cal R}_{\alpha}$, we have
\begin{equation}\label{eq:I2}
 I_2 \leq \frac{1}{H^{\alpha-1}}O\left(\int_x^\infty\Pr[X_1>u]du\right).
\end{equation}
Similarly, by defining $X^{\delta}_n=X_n+(3\delta)/4$ with the
partial sum $S_n^{\delta}=\sum_1^nX_i^{\delta}$ and noting that
$S^{\delta}_{Hx}\le \sup_{n\ge 0}S^{\delta}_n$, we obtain
\begin{equation}\label{eq:hI1}
 I_1 \leq \Pr \left[\sup_n S^{\delta}_n > \frac{\delta Hx}{4} \right]\leq
 \frac{1}{H^{\alpha-1}}O\left(\int_x^\infty\Pr[X_1>u]du\right).
\end{equation}
Combining (\ref{eq:I2}) and (\ref{eq:hI1}) completes the
proof.\qed
\end{proof}

\noindent{\bf Proof of Theorem \ref{theorem:isolate}:} By the
classical result of \cite{LOY62} (see also Chapter~2.2 of
\cite{BAB94}),
 we have
 \begin{equation*}
       W^{(k)} \eqd \left( W^{(k)}(T_{-1})+ B^{(k)}_{-1} + c_kT_{-1}
       \right)^+,
 \end{equation*}
 where $W^{(k)}(T_{-1})$ is the stationary workload observed at
 the moment $T_{-1}$. Furthermore,
 $W^{(k)}(T_{-1})\eqd \sup_{n\ge 0}S_n$, with $S_n = \sum_{i=1}^{n}X_i, n\ge 1, S_0=0$
  and $X_i \eqdef  B_{i}^{(k)} - c_k
(T_{i}-T_{i-1})$. Next, observe that for $x>0$
 \begin{align}\label{eq:isototal}
  \Pr[W^{(k)}(T_{-1})>x] &=\Pr\left[\sup_{n\ge 1}S_n>x \right]
 \le \Pr\left[ \sup_{n\leq Hx}S_{n}>x \right]+ \Pr\left[ \sup_{n\geq Hx}S_{n}>x
  \right]\nonumber \\&\le \Pr \left[\sup_{n\ge 1}\underline{S}_n^\epsilon>
\delta x \right]+\Pr \left[\sup_{1\le n\le
Hx}\overline{S}_n^\epsilon>(1-\delta)x \right]+ \Pr\left[
\sup_{n\geq Hx}S_{n}>x
  \right] \nonumber\\
&\quad =I_1(x)+I_2(x)+I_3(x),
 \end{align}
where ${\overline{X}} _i^\epsilon=X_i\ind\{X_i>\epsilon
 x\}$, ${\underbar{X}} _i^\epsilon=X_i\ind\{X_i\leq \epsilon
 x\}$,  and ${\overline{S}}_{n}^{\epsilon} =\sum_{i=1}^{n} \left( {\overline{X}}
 _i^\epsilon+ \expect[X_i]+\delta
 \right)$, ${\underline{S}}_{n}^{\epsilon} =\sum_{i=1}^{n} \left( {\underline{X}}
 _i^\epsilon- \expect[X_i] -\delta
 \right)$ are defined for some $\epsilon>0, |\expect[X_1]|>\delta>0$.

First, let us prove an \emph{upper bound } for
(\ref{eq:isototal}).
 By Lemma 3.2(i) in \cite{JEMO03}, for any $\beta>0$, there
exists $\epsilon>0$ such that
\begin{equation}\label{eq:isoI1}
I_1(x)=o(x^{-\beta}).
\end{equation}
 Furthermore,
define $N_k=\sum_{i=1}^{Hx}\ind\{\overline{X}_i^\epsilon>0\}, 0\le
k\le m$; note that ${\overline{X}} _i^\epsilon$ depends on the
class index $k$ since $X_i=B_{i}^{(k)} - c_k (T_{i}-T_{i-1})$. To
simplify the notation, we assume that $Hx$ is an integer. Now,
$\Pr[ N_k \geq 2 ]$ is upper bounded by
\begin{align*}
   {{Hx}\choose{1}} \Pr \left[B^{(k)}> \epsilon
  x\right]{{m-1}\choose {1}}\Pr\left[B> \epsilon
  x\right]+ {{Hx}\choose{2}} \left(\Pr \left[B^{(k)}> \epsilon
  x\right]\right)^2=o\left( x\left(\Pr[B> x]\right)^{k+1} \right).
\end{align*}
In the preceding expression, the first term bounds the sum of
probabilities
$\Pr[\overline{X}_i^\epsilon>0,\overline{X}_j^\epsilon>0]$ for all
indices $1\le |i-j|\le m$ (note that in this case
$\overline{X}_i^\epsilon$ and $\overline{X}_j^\epsilon$ are
dependent); the second term provides a bound on the corresponding
sum when $|i-j|> m$, using the fact that $\overline{X}_i^\epsilon$
and $\overline{X}_j^\epsilon$ are independent. Therefore,
 \begin{align}\label{eq:isoI2}
  I_2(x) &\leq  \Pr \left[\sup_{0\le n\le
Hx}\overline{S}_n^\epsilon>(1-\delta)x , N_k=1 \right]+ \Pr[N_k \geq 2] \nonumber \\
&\le \sum_{n=1}^{Hx} \Pr\left[
{\overline{X}}_i^{\epsilon}+n(\expect[X_1]+\delta)>(1-\delta)x
\right]+o\left(x
\left(\Pr[B> x]\right)^{k+1} \right)\nonumber\\
&\leq \int_{0}^{\infty}\Pr[X_1>(1-\delta)x+u |\expect[X_1]-\delta
| ]du + o\left(x\left(
\Pr[B> x]\right)^{k+1} \right) \nonumber\\
&\sim  \frac{1}{|\expect[X_1]-\delta
|}\int_{(1-\delta)x}^{\infty}\Pr[X_1>u]du.
 \end{align}
 The estimate for $I_3(x)$ follows from Lemma \ref{lemma:geqHx}.
 Using this estimate, (\ref{eq:isototal}), (\ref{eq:isoI1}), (\ref{eq:isoI2}) and
passing $\delta, \epsilon \to 0$,  $H\to \infty$, we
 obtain the upper bound.

Next, we prove the \textit{lower bound} for (\ref{eq:isototal})
\begin{align}
\Pr[W^{(k)}(T_{-1})>x]&\ge \Pr\left[\sup_{1\le n\le Hx}S_n>x \right]
\ge \Pr\left[\sup_{1\le n\le Hx}\overline{S}_n^\epsilon>x\right] \ge
\Pr\left[\sup_{1\le n\le Hx}\overline{S}_n^\epsilon>x, N_k=1
\right]\nonumber\\
&= \sum_{n=1}^{Hx} \Pr\left[ {\overline{X}}_i^{\epsilon}+n(\expect
X_1+\delta)>x\right] \ge \int_1^{Hx}\Pr \left[X_1>x+u|\expect
X_1-\delta| \right]du, \nonumber
\end{align}
which by passing $x\rightarrow \infty$, using regular variation, and
then passing $\delta \rightarrow 0$, results in
\begin{align}
\Pr[W^{(k)}(T_{-1})>x]& \gtrsim \frac{1}{ c_k \expect[ T]-\expect
  [B^{(k)}]}\int_{x}^{\infty} \Pr \left[B^{(k)} > u
  ]\right]du.\label{eq:ubdt3}
\end{align}

Finally, for any $0<\epsilon<1$, we have
\begin{align}
  \Pr\left[W^{(k)}>x \right] & = \Pr\left[ \left( W^{(k)}(T_{-1})+ B^{(k)}_{-1} + c_kT_{-1}
       \right)^+ >x  \right]\nonumber \\
       &\leq \Pr\left[ W^{(k)}(T_{-1}) > (1-\epsilon)x  \right]  + \Pr\left[ B^{(k)} > \epsilon x  \right] \nonumber,
\end{align}
which, by (\ref{eq:ubdt3}), and then passing $\epsilon \to 0$,
yields
 \begin{equation}\label{eq:hp1}
      \Pr\left[W^{(k)}>x \right]  \lesssim  \Pr\left[ W^{(k)}(T_{-1}) > x
      \right].
 \end{equation}
 Also, since $W^{(k)}(T_{-1})$ is heavy-tailed and independent of $T_{-1}$, we
 obtain
 \begin{align}\label{eq:hp2}
  \Pr\left[W^{(k)}>x \right] & \geq \Pr\left[  W^{(k)}(T_{-1}) + c_kT_{-1}
        >x  \right]\sim \Pr\left[ W^{(k)}(T_{-1}) > x  \right].
\end{align}
Thus, (\ref{eq:hp1}) and (\ref{eq:hp2}) imply
 \begin{equation}
      \Pr\left[W^{(k)}>x \right]  \sim  \Pr\left[ W^{(k)}(T_{-1}) > x
      \right],
 \end{equation}
which, combined with (\ref{eq:ubdt3}), completes the proof of the
first asymptotics. The second asymptotic relationship of the
theorem is implied by Lemma~\ref{lemma:orderstat2}.\qed

\subsubsection{Static Priority} \label{ss:sp} In this subsection, we assume that there is only one server with
capacity $c$ and that the $m+1$ classes are served jointly with a
preemptive static priority (SP) discipline between classes.  Suppose
that the priorities of the classes are assigned in a decreasing
order of the class index $k$, $0\leq k\leq m$, i.e., class $k$
receives service only if classes $i, k+1\leq i\leq m$ are empty.
Denote by $W_0^{(k)}$ the stationary workload of class $k$ observed
at arrival point $T_0$. Let $\mu^{(k)} \eqdef \sum_{i=k}^{m}\expect
\left[B^{(i)} \right]$ and note that $\mu^{(0)}=\expect[B]$.
\begin{theorem}\label{theorem:static}
If $\Pr[B>x]=l(x)/x^{\alpha} \in {\cal R}_{\alpha},\, \alpha > 1$
and $\expect[B]< c \expect[T]$, then,  as $ x$ $ \to \infty$,
\begin{align*}
  \Pr\left[ W_0^{(k)} >x  \right]&\sim
   \frac{1}{c \expect[T]-
  \mu^{(k)} }\int_{x}^{\infty} \Pr \left[B^{(k)} > u \right]du\\
&\sim \frac{1}{(k+1)(c \expect[T]-
  \mu^{(k)})}{m\choose k}\frac{l(x)^{k+1}}{x^{\alpha k+\alpha
  -1}}.
\end{align*}
\end{theorem}
\begin{remark}
This result shows that the distribution of the workload $W_0^{(k)}$
behaves asymptotically as if class $k$ were served in isolation by a
system with capacity reduced by the mean job sizes of classes with
indices greater than $k$, which indicates a similar phenomenon as in
Theorem \ref{theorem:isolate}. Thus, our SP scheduling with
comparison splitter should approximate SRPT well.
\end{remark}
\begin{proof} Let $W^{(k)}(T_{n})$ be the stationary workload of class $k$
jobs at time $T_n$. First, we establish an \emph{upper bound}. For
$0\leq k\leq m$, we group all the arrivals of classes $k,\cdots,
m$ into a new class with the highest priority, while all the other
classes remain the same. The workload of the new class is denoted
as $\hat{W}^{(k)}(T_n)$, where
 $\hat{W}^{(k)}(T_n) \eqdef \sum_{i=k}^{m}W^{(i)}(T_n)$  and $\hat{W}^{(k)}_0$ represents a variable that is equal in
 distribution to $\hat{W}^{(k)}(T_n)$.
Clearly,
\begin{equation}\label{eq:boundW}
 W^{(k)}(T_n) \leq  \hat{W}^{(k)}(T_n),
\end{equation}
where the workload recursion for the new class satisfies
 \begin{align*}
  \hat{W}^{(k)}(T_{n+1}) &=
  \left( \hat{W}^{(k)}(T_{n})+ \sum_{i=k}^{m}B^{(i)}_{n+1} -c
  (T_{n+1}-T_n)
  \right) ^+.
 \end{align*}

Now, by Lemma \ref{lemma:orderstat2}, it is easy to see that, as
$x\to \infty$,
\begin{equation*}
 \Pr\left[\sum_{i=k}^{m}B^{(i)}_{n+1}>x \right] \sim \Pr \left[B^{(k)}_{n+1}>x \right],
\end{equation*}
and, using the same argument as in the proof of the upper bound in
Theorem \ref{theorem:isolate}, we obtain
\begin{equation}\label{eq:accumulate}
  \Pr\left[ \hat{W}_0^{(k)} >x  \right]\sim
   \frac{1}{c \expect[T]-
  \mu^{(k)} }\int_{x}^{\infty} \Pr \left[B^{(k)} > u \right]du,
\end{equation}
which, by (\ref{eq:boundW}), yields
\begin{equation}\label{eq:SPupper}
  \Pr\left[ W_0^{(k)} >x  \right]\lesssim
   \frac{1}{c \expect[T]-
  \mu^{(k)} }\int_{x}^{\infty} \Pr \left[B^{(k)} > u \right]du.
\end{equation}

 Next, we prove a \emph{lower bound}. For $\epsilon >0$ and $k<m$, we
have
\begin{align*}
  \Pr\left[W_0^{(k)}>x \right] &\geq \Pr \left[W_0^{(k)}>x,  \hat{W}_0^{(k+1)}\leq \epsilon  x \right]\geq \Pr \left[\hat{W}_0^{(k)}> (1+\epsilon)x, \hat{W}_0^{(k+1)}\leq \epsilon
  x \right]\nonumber \\
  &\geq \Pr \left[\hat{W}_0^{(k)}> (1+\epsilon)x \right] -\Pr \left[ \hat{W}_0^{(k+1)}> \epsilon x
  \right].
 \end{align*}
Using the same argument as for (\ref{eq:accumulate}) and passing
$\epsilon \to 0$ in the preceding inequality imply
\begin{equation*}
  \Pr\left[W_0^{(k)}>x \right] \gtrsim \frac{1}{c \expect[T]-
  \mu^{(k)} }\int_{x}^{\infty} \Pr \left[B^{(k)} > u \right]du.
  \end{equation*}
The same asymptotic inequality can be easily shown to hold for
$k=m$. This inequality, combined with (\ref{eq:SPupper}), completes
the proof of the first asymptotic relationship in Theorem
\ref{theorem:static}. The second asymptotics follows directly from
Lemma \ref{lemma:orderstat2}.\qed
\end{proof}

\section{Conclusion}\label{sec:con}
 We show in Theorem~\ref{theorem:main} that the
medium size heavy-tailed jobs can have asymptotically much shorter
sojourn times under SRPT than under PS/FBPS scheduling disciplines.
Furthermore, the asymptotic performance of SRPT is uniformly good
for the smaller as well as for the larger jobs, which implies that
the performance gains of smaller jobs with SRPT, compared to
PS/FBPS, are not achieved at the expense of larger jobs. Hence, in
this asymptotic heavy-tailed context, SRPT is both efficient and
fair, which complements similar findings obtained using the mean
value analysis.

However, as early as in the paper by \cite{SCMI66}, it was observed
that SRPT may be difficult to implement because of its complicated
preemptive nature that requires keeping track of the remaining
processing times for all the jobs in the queue. Thus, it is natural
to consider threshold-based static priority (SP) disciplines to
approximate SRPT, as suggested originally by \cite{SCMI66}, which
was then followed by a considerable number of later studies.
However, the main drawback of selecting static thresholds in
practice is that the real world traffic is often nonstationary,
highly  correlated, bursty, etc.

Our second main contribution in this paper is the design of a
scalable (low complexity) and adaptive comparison scheduling
approximation to SRPT. The good performance of our comparison
scheduler is demonstrated using our asymptotic queueing analysis
under the heavy-tailed service requirements; additional verification
of this scheduling algorithm was done by \cite{JKTSIG07} via
simulations. We also discuss refinements of our mechanism that, at
the expense of a small additional complexity, improve the accuracy
of job classification for correlated arrivals and highly
concentrated service distributions.

Finally, we would like to point out that, in addition to the static
priority discipline analyzed in our paper, it may also be
interesting to analyze the performance of our splitting mechanism
for other disciplines,  such as generalized processor sharing in
\cite{BBJ00c}, weighted fair queueing in \cite{BJC06}, and
hierarchical processor sharing.

\section*{Appendix}
 \subsubsection*{Proof of Lemma \ref{lemma:orderstat2}}
\label{sec:proof} Since the case $m=0$ is immediate, we assume that
$m\ge 1$. First, we show that the second asymptotics in
(\ref{eq:two}) holds assuming that $\{X_i\}_{0\le i\le m}$ are
continuous. In this case, we have $\Pr[X_i=X_j]=0, i\neq j$ and,
thus
\begin{align*}
\Pr[X_0>x, X_0< \tilde{X}_k]&={m\choose
k}\Pr\left[X_0>x,X_0\le\min_{1\le i\le
k}X_i\right]=\frac{1}{k+1}{m \choose k}\Pr\left[\min_{0\le i\le
k}X_i>x\right]\\&=\frac{1}{k+1}{m \choose k}\Pr[X_0>x]^{k+1}.
\end{align*}
Next, the first asymptotics in (\ref{eq:two}) is implied by the
preceding analysis and the following identity
\begin{align*}
\Pr[X_0>x, \tilde{X}_{k+1}\le X_0<
\tilde{X}_k]=\Pr[X_0>x,X_0<\tilde{X}_k]-\Pr[X_0>x,X_0<\tilde{X}_{k+1}].
\end{align*}

If $\{X_i\}_{0\le i\le m}$ are not continuous but in $\cal L$,
(\ref{eq:two}) still holds asymptotically. This claim will follow
from the preceding arguments if we show that for $X_i\in {\cal
L}$, as $x\rightarrow \infty$,
\begin{equation}
\label{eq:ht} \Pr[X_n>X_{n-1}\cdots>X_0>x]\sim \Pr[X_n\ge
X_{n-1}\cdots\ge X_0>x].
\end{equation}
Since $X_i\in {\cal L}$, it is enough to prove the preceding
relationship for $x\in \nat$.  Our proof starts with $n=1$,
\begin{align}\label{eq:induction1}
\Pr[X_0>x, X_0\le X_1]=\Pr[X_0>x, X_0<X_1]+\Pr[X_0>x, X_0=X_1].
\end{align}
Furthermore, for any $\epsilon>0$ and $x$ large,
\begin{align}
 \Pr[X_0>x, X_0=X_1]&=\sum_{y=x}^\infty\Pr[y<X_0\le
y+1, X_0=X_1, y<X_1\le y+1]\nonumber \\&\le
\sum_{y=x}^\infty(\Pr[y<X_0\le y+1])^2\nonumber\\
&=\sum_{y=x}^\infty\Pr[y<X_0\le y+1]
\frac{\Pr[X_0>y]}{\Pr[X_0>y+1]}\Pr[X_0>y+1]\nonumber\\
& \quad -\sum_{y=x}^\infty\Pr[y<X_0\le y+1]\Pr[X_0>y+1]\nonumber\\
&\le \epsilon\sum_{y=x}^\infty\Pr[y<X_0\le y+1]\Pr[X_0>y+1]
\label{eq:htp}\\&\leq \epsilon (\Pr[X_0>x])^2,\nonumber
\end{align}
where the last inequality is implied by the monotonicity of
$\Pr[X_0>y]$ and (\ref{eq:htp}) follows from $X_0\in{\cal L}$
since for any $\epsilon>0$, we can choose $x_0$ such that for
$y>x\ge x_0$,
$$\frac{\Pr[X_0>y]}{\Pr[X_0>y+1]}\le 1+\epsilon.$$
Combining (\ref{eq:induction1}), (\ref{eq:htp}), using the fact that
$\Pr[X_0>x, X_0<X_1]$ is of the same order as $(\Pr[X_0>x])^2$, and
passing $\epsilon \to 0$, yield the proof for $n=1$. Now,
for $n\geq 2$, we have
\begin{align*}
\Pr[X_{n}\geq X_{n-1} & \cdots\geq
X_0>x]=\Pr[X_{n}>X_{n-1}\geq\cdots\geq X_0>x]+\Pr[X_{n}=X_{n-1} \geq
\cdots\geq
X_0>x]\\
&\leq \Pr[X_{n}> X_{n-1} \geq \cdots \geq
X_0>x]+\Pr[X_{n}=X_{n-1}>x]\Pr[X_{n-2}\geq \cdots\geq X_0>x]\\
&= \Pr[X_{n}> X_{n-1} \geq \cdots \geq  X_0 >
x]+o\left(\Pr[X_0>x]^{n+1} \right),
\end{align*}
and by repeating the preceding procedure $n-1$ more times, we obtain
\begin{align*}
\Pr[X_{n}\geq X_{n-1}  \cdots\geq X_0>x] = \Pr[X_{n}> X_{n-1} >
\cdots
>  X_0 > x]+o\left(\Pr[X_0>x]^{n+1} \right).
\end{align*}
Noting that $\Pr[X_{n}> X_{n-1} > \cdots
>  X_0 > x]$ is of the same order as $\Pr[X_0>x]^{n+1}$ and
$\Pr[X_{n}\geq X_{n-1}  \cdots\geq X_0>x] \geq \Pr[X_{n}> X_{n-1} >
\cdots
>  X_0 > x]$, we finish the proof.
 \qed
 \subsubsection*{Proof of Lemma \ref{lemma:compare}}
Let $R_{10}(t)$ and $R_{20}(t)$ be the remaining service times at
time $t\geq 0$ for the labeled customer that arrives at $T_0$
under processes $\{(T_i,B_{1i}) \}_{i> -\infty}$ and
$\{(T_i,B_{2i}) \}_{i> -\infty}$, respectively. By the same
notion, we define $W_1(t)$ and $W_2(t)$ to be the workloads at
time $t$ in these two queues that need to be finished before the
labeled customer can start receiving its service. In order to
justify $V_1\leq V_2$, it is enough to prove that $R_{10}(t)\leq
R_{20}(t),t\geq 0$.

We use induction to prove the result and denote $\max (x,0)$ by
$x^+$. First, if $R_{10}(T_i)\le R_{20}(T_i)$ and $W_1(T_i+) \leq
W_2(T_i+)$, we have
\begin{align}\label{eq:lemma1p2}
 W_1(t) = \left( W_1(T_i+) - (t-T_i) \right)^+  &\le \left( W_2(T_i+) - (t-T_i) \right)^+ =W_2(t)  \nonumber\\
R_{10}(t)=R_{10}(T_i)-(t-T_i-W_1(T_i+))^+ &\le
R_{20}(T_i)-(t-T_i-W_2(T_i+))^+ = R_{20}(t)
\end{align}
 for $T_i\le t< T_{i+1}$. Note that $W_j(T_i+)$ and $W_j(T_i-)$ denote the right- and left-hand limits
 of $W_j(t)$ at $T_i$, respectively; i.e., the times right after and before the arrival at $T_i$.  Hence, it
is enough to prove that,  all the customers arriving at $T_i$,
$T_0\le T_i\le V_1$, see $R_{10}(T_i) \leq R_{20}(T_i)$ and
$W_1(T_i+) \leq W_2(T_i+)$ immediately after their arrival.

For the arrival at time $T_0$, the claim is obviously correct. Now,
assuming that the result holds for $i=n$, we proceed to prove it for
$i=n+1$. Based on the hypothesis, (\ref{eq:lemma1p2}) implies
$R_{10}(T_{n+1}) \leq R_{20}(T_{n+1})$ and $W_1(T_{n+1}-) \leq
W_2(T_{n+1}-)$ at the time immediately before $T_{n+1}$. Next, at
time $T_{n+1}$, if $B_{1(n+1)}=0<B_{2(n+1)}$, then, we have
\begin{equation*}
W_1(T_{n+1}+)=W_1(T_{n+1}-)\le W_2(T_{n+1}-)+B_{2(n+1)}\ind
\left\{B_{2(n+1)}<R_{20}(T_{n+1}) \right\}=W_2(T_{n+1}+),
\end{equation*}
since $W_1(T_{n+1}-) \leq W_2(T_{n+1}-)$.

The case  $B_{1(n+1)}=B_{2(n+1)}=B_{n+1}$  results in the following
three different scenarios:
\begin{enumerate}
\item[1)] If $B_{n+1}<R_{10}(T_{n+1})$, then
$$W_1(T_{n+1}+) =
W_1(T_{n+1}-)+ B_{n+1} \leq W_2(T_{n+1}-)+ B_{n+1}=W_2(T_{n+1}+),$$
since  $R_{10}(T_{n+1}) \leq R_{20}(T_{n+1})$ by induction
hypothesis.

\item[2)] If $B_{n+1}>R_{20}(T_{n+1})$, then
$$W_1(T_{n+1}+) =
W_1(T_{n+1}-)\leq W_2(T_{n+1}-)=W_2(T_{n+1}+).$$ \item[3)] If
$R_{10}(T_{n+1})\leq B_{n+1} \leq R_{20}(T_{n+1})$, then
$$W_1(T_{n+1}+) =
W_1(T_{n+1}-)\leq W_2(T_{n+1}-)+ B_{n+1}=W_2(T_{n+1}+).$$
\end{enumerate}
Therefore, the result holds for $i=n+1$, which completes the
induction, and implies that $V_2 \geq V_1$. \qed
\subsubsection*{Proof of Lemma \ref{lemma:refined}}
Without loss of generality we  assume that $b_1>\cdots>b_{\nu}$ and
$\min \{p_k\}_{1\leq i\leq \nu}>0$. Define $q_k \eqdef
\sum_{i=1}^{k}p_i$, $1\leq k \leq \nu$ with $q_0=0$ and choose
$m>\min \{1/p_k\}_{1\leq k \leq \nu} $.
 When $B_i=b_k$, we say $B_i$ is
routed into the \emph{right class} if $B_i$ is either in class
$\lfloor m q_{k-1} \rfloor$  or in class $\lceil m q_{k-1}-1 \rceil$
(note that if $mq_{k-1} \notin \nat$, then $\lfloor m q_{k-1}
\rfloor=\lceil m q_{k-1}-1 \rceil$).
 The
condition $m >\min \{ 1/p_k \}_{1\leq k\leq \nu}$ guarantees that if
$B_i \neq B_j$, then, the corresponding right classes for $B_i$ and
$B_j$ are different since $mp_k>1$ for all $1\leq k \leq \nu$.

 First,  since both
$\{O_i\}$ and $\{S_i\}$ are random, we construct a deterministic
sequence $\{d_i\}_{1\leq i\leq n}$ for comparison purposes as
follows: $d_i=b_k, \lfloor n q_{k-1} \rfloor+1 \leq  i \leq \lfloor
n q_{k} \rfloor$. Then,
\begin{align}\label{eq:refine5}
\Pr[ \eta(n) > \epsilon ]&\leq \Pr\left[\sum_{i=1}^{n} \ind\{O_i
\neq d_i\}
> \frac{\epsilon}{2} n \right]
       +\Pr \left[\sum_{i=1}^{n} \ind\{S_i \neq d_i\} > \frac{\epsilon}{2} n \right] \nonumber\\
       & \eqdef I_1+I_2.
\end{align}
For $I_2$, applying the union bound, we can easily prove that, for
some $H, \xi>0$,
\begin{equation}\label{eq:refine6}
I_2\leq H e^{-\xi
   n}. \end{equation} Therefore,  we only need to prove
that $I_1\leq He^{-\xi \min(l, n)}$, where $H, \xi$ may be
different from the ones chosen in (\ref{eq:refine6}).

Next, in order to evaluate $I_1$,  we denote the event $\mathcal
{E}_i=\{ \text {$B_i$ is not in the right class} \}$  and prove that
there exists $H, \xi>0$, such that as $n\to \infty$,
\begin{equation}\label{eq:refine2}
  \max_{1\leq i \leq \nu }\Pr[\mathcal {E}_i]\leq H e^{-\xi l}.
\end{equation}
To this end, if $\nu=1$, it is obvious that $\Pr[  \mathcal
{E}_i]=0$ for all $i$;
 if $\nu \geq 2$, noting that $\Pr[\mathcal {E}_i, B_i =
  b_1]=0$,  we have
\begin{align}\label{eq:refine3}
  &\Pr[\mathcal {E}_i]
  =\sum_{k=1}^{\nu-1}  \Pr[\mathcal {E}_i,  B_i =
  b_{k+1}],
\end{align}
where $\Pr[\mathcal {E}_i,   B_i =
  b_{k+1}]$ is upper bounded by
\begin{align}\label{eq:refineI12}
&\Pr \left[ \left\{
              \tilde{B}_{( i, \lfloor m q_k +1\rfloor l ) } \leq b_{k+1} < \tilde{B}_{(i,
              \lceil
   m q_k -1 \rceil l )}   \right\}^C \right]\nonumber\\
   &\;\; \leq \Pr \left[  b_{k+1}<
             \tilde{B}_{( i, \lfloor m q_k+1 \rfloor l ) }  \right]+ \Pr \left[
             b_{k+1} \geq
             \tilde{B}_{( i, \lceil m q_k-1 \rceil l ) }  \right]\nonumber \\
             &\;\; = \Pr \left[ \sum_{i=1}^{ml}\ind\{B_i<b_{k+1}\} >  \lfloor mq_k +1\rfloor
 l \right]+\Pr \left[ \sum_{i=1}^{ml}\ind\{B_i<b_{k+1}\} \leq  \lceil mq_k -1
             \rceil l
 \right].
\end{align}
By noting that $\expect[\ind\{B_i<b_{k+1}\}]=q_k, 1\leq k\leq
\nu-1$,  and  using the large deviation results with the condition
$\lfloor mq_k +1\rfloor > m q_k
> \lceil mq_k -1
             \rceil $, we obtain that for all $1\leq k \leq
\nu-1$ and some $H, \xi>0$, the righthand side of
(\ref{eq:refineI12}) is further bounded by $H e^{-\xi
   l}$.
By substituting this upper bound for (\ref{eq:refineI12}) into
(\ref{eq:refine3}), we prove (\ref{eq:refine2}), and therefore, the
total number of jobs
$$
N_{\epsilon} \eqdef  \sum_{i=1}^{n}  \ind\{ \mathcal{E}_i \}
$$ that are not in the right
classes satisfies, for $0< \delta <1$ and some $H_{\delta}, \xi>0$,
\begin{align}\label{eq:refine4}
   \Pr[N_{\epsilon}> \delta n]&=
      \Pr\left[ \sum_{i=1}^{n}  \ind\{ \mathcal{E}_i \}  >
      \delta n
      \right]\leq  \frac{\expect\left [\sum_{i=1}^{n}  \ind\{ \mathcal{E}_i \}  \right]}{\delta n}
  \leq  H_{\delta} e^{-\xi
   l}.
\end{align}

Now, we continue with evaluating $I_1$. Since
\begin{align}\label{eq:refine7}
I_1&\leq \sum_{k=1}^{\nu} \Pr\left[\sum_{i=1}^{n} \ind \{ O_i \neq
d_i, O_i = b_k \}
> \frac{\epsilon n}{2\nu} \right],
\end{align}
we only need to show that for each $1\leq k\leq \nu$ and some
$H_{\epsilon}, \xi_{\epsilon}>0$,
$$
\Pr\left[\sum_{i=1}^{n} \ind \{ O_i \neq d_i, O_i = b_k \}
> \frac{\epsilon n}{2\nu} \right]\leq H_{\epsilon} e^{-\xi_{\epsilon}
  \min( l,n)}.
$$
To this end, we define $E_n^{(k)}\eqdef \sum_{i=1}^{n} \ind\{ O_i
\neq d_i, O_i=b_k \}$ and
 denote by $N_k, 1\leq k\leq \nu$ the total
number of jobs of size $b_k$ and by $N_k^r$ the total number of jobs
of size $b_k$ that are routed into the right class with
$N_0=N_0^r=0$.  Obviously, by the definition of $N_{\epsilon}$, we
have $\left| \sum_{j=0}^{k} (N_j^r - N_j) \right|\leq N_{\epsilon}$
for $1\leq k\leq \nu$. Now, we claim that, for $1\leq k\leq \nu$,
\begin{equation}\label{eq:refine8}
E_n^{(k)}= \sum_{i=1}^{n} \ind\{ O_i \neq d_i, O_i=b_k \} \leq
\left| \sum_{j=0}^{k-1} N_{j}^r- \lfloor n q_{k-1} +1 \rfloor
 \right |
   + \left| \sum_{j=0}^{k} N_{j}^r- \lfloor n q_k \rfloor \right|  +
  2 N_{\epsilon}.
  \end{equation}

 In order to prove (\ref{eq:refine8}), we define $\mathcal{R}_k \subset \{1,2,\cdots, n\}$ to be the set
 of all the indices of the jobs in
 $\{O_i\}_{1\leq i \leq \nu}$ that are routed to the right classes for job size
 $b_k$. Now, if there is no element $i$ of $\mathcal{R}_k $ such that $O_i = b_k$, then the total number of
 jobs of size $b_k$ in $\{O_i\}_{1\leq i \leq \nu}$ is bounded by $N_{\epsilon}$ since
 none of
 the jobs of size $b_k$ are in the
 right classes. Thus, in this case we obtain
$$
   E_n^{(k)} \leq \sum_{i=1}^{n} \ind\{O_i = b_k \} \leq
   N_{\epsilon}.
$$
 Next, if
$\mathcal{R}_k $ contains at least one index $i$ such that $O_i =
b_k$, we can always define $\tau_k = \min\{i \in \mathcal{R}_k : O_i
= b_k \}$ and $\sigma_k = \max\{i \in \mathcal{R}_k : O_i = b_k \}$.
Then, let $\mathcal {A}=\{i \in {\cal N} \mid \tau_k \leq i \leq
\sigma_k \}$
  and $\mathcal {B}=\{i \in {\cal N} \mid \lfloor n q_{k-1}+1 \rfloor \leq i \leq \lfloor n q_{k}\rfloor \}$.
  It is easy to see that all the indices in $\mathcal {A}$ but not in $\mathcal {B}$
   are contributing to $E_n^{(k)}$ since $d_i \neq b_k$ for $i \in \mathcal{A}\backslash \mathcal{B}$, and therefore,
  $$E_n^{(k)} \leq \mid \mathcal {A}\backslash  \mathcal {B}\mid +
  N_{\epsilon},
  $$
 where $N_{\epsilon}$ contains all the errors $\ind\{ O_i \neq
d_i, O_i=b_k \}$ for $i \notin \mathcal{R}_k$. Here,
``$\backslash$'' represents set difference operation and $\mid \cdot
\mid$ denotes the cardinality of a set.
   To compute the cardinality of the preceding set difference, we have the following four different
  scenarios.
  \begin{itemize}
  \item if $\tau_k\leq \lfloor n q_{k-1}+1 \rfloor $ and $\sigma_k
    \leq \lfloor n q_{k}\rfloor$, then
$\mid \mathcal {A}\backslash \mathcal {B} \mid$ is upper bounded by
$\lfloor n q_{k-1} +1 \rfloor-\tau_k$,
 which, by noting that $\sum_{j=0}^{k-1} N_{j}^r<  \tau_k$,  results
 in
 $$E_n^{(k)} \leq \lfloor n q_{k-1} +1 \rfloor-\tau_k+N_\epsilon\leq  \left| \sum_{j=0}^{k-1} N_{j}^r- \lfloor n q_{k-1} +1 \rfloor
 \right |+N_\epsilon;$$
 \item if $ \tau_k \geq \lfloor n q_{k-1}+1 \rfloor$ and $\sigma_k \geq \lfloor n q_{k}\rfloor   $, then
$\mid \mathcal {A} \backslash \mathcal {B} \mid$ is upper bounded by
$\sigma_k-\lfloor n q_{k} \rfloor$.  By
  noting that $\sum_{j=0}^{k} N_{j}^r+N_{\epsilon}\geq
  \sigma_k$, we obtain
 $$E_n^{(k)} \leq \sigma_k-\lfloor n q_{k} \rfloor +N_\epsilon\leq  \sum_{j=0}^{k} N_{j}^r+N_{\epsilon}- \lfloor n q_{k} \rfloor+N_\epsilon;$$
    \item  if $\lfloor n q_{k-1}+1 \rfloor \leq \tau_k \leq \sigma_k
    \leq \lfloor n q_{k}\rfloor$, then, $\mid \mathcal {A} \backslash \mathcal {B} \mid=0$ and $E_n^{(k)}$ is upper bounded by the total number of jobs that
    are not in the right classes $N_{\epsilon}$;
 \item if $\tau_k  <  \lfloor n q_{k-1}+1 \rfloor  \leq  \lfloor n q_{k}\rfloor < \sigma_k$,
 then, we obtain $\mid \mathcal {A} \backslash \mathcal {B} \mid \leq \lfloor n q_{k-1} +1 \rfloor - \tau_k + \sigma_k
- \lfloor n q_{k} \rfloor$, which, by
 noting that $\sigma_k \leq \sum_{j=0}^{k} N_{j}^r+N_\epsilon$ and $\sum_{j=0}^{k-1} N_{j}^r<  \tau_k$,
yields
\begin{align}
E_n^{(k)} &\leq  \lfloor n q_{k-1} +1 \rfloor - \tau_k + \sigma_k -
\lfloor n
q_{k} \rfloor+N_{\epsilon} \nonumber\\
&\leq \left| \sum_{j=0}^{k-1} N_{j}^r- \lfloor n q_{k-1} +1
\rfloor
 \right|
   +  \sum_{j=0}^{k} N_{j}^r + N_{\epsilon} - \lfloor n q_k \rfloor +
   N_{\epsilon}.\nonumber \end{align}
  \end{itemize}
Therefore, by the above arguments, we prove the claim in
(\ref{eq:refine8}).

 Next, using (\ref{eq:refine8}),  for $1\leq
k\leq \nu$, we derive
\begin{align}
&\Pr\left[\sum_{i=1}^{n} \ind\{O_i \neq d_i, O_i=b_k\}
> \frac{\epsilon n}{2\nu} \right]\nonumber\\
 &\leq \Pr \left[  \left| \sum_{j=0}^{k-1} N_{j}^r- \lfloor n q_{k-1} +1 \rfloor
 \right |
   + \left| \sum_{j=0}^{k} N_{j}^r- \lfloor n q_k \rfloor \right|  + 2N_{\epsilon} > \frac{\epsilon n}{2 \nu}
   \right]\nonumber \\
&\leq \Pr \left[  \left| \sum_{j=0}^{k-1} N_{j}- \lfloor n q_{k-1}
+1 \rfloor
 \right |
   + \left| \sum_{j=0}^{k} N_{j}- \lfloor n q_k \rfloor \right|  + 4N_{\epsilon} > \frac{\epsilon n}{2 \nu}
   \right]\nonumber \\
   &\leq \Pr \left[ \left| \sum_{j=0}^{k-1} N_{j}- \lfloor n q_{k-1}
+1 \rfloor
 \right | > \frac{\epsilon n}{6 \nu} \right]+ \Pr\left[N_{\epsilon}> \frac{\epsilon n}{24 \nu} \right] \nonumber\\
& \quad  + \Pr \left[ \left| \sum_{j=0}^{k} N_{j}- \lfloor n q_k
\rfloor \right|
     > \frac{\epsilon n}{6 \nu} \right],
\nonumber
\end{align}
which, by noting that $\expect[N_j]=np_j$ for $1\leq j\leq \nu$ and
using Chernoff bound, (\ref{eq:refine7}) and (\ref{eq:refine4}),
implies that $I_1\leq He^{-\xi \min(l, n)}$ for some $H, \xi>0$.
Combining this bound, (\ref{eq:refine5}) and (\ref{eq:refine6}), we
complete the proof.
 \qed
\small
\bibliographystyle{apalike}

\end{document}